\documentclass[a4paper, 11pt]{article}
\usepackage{booktabs} 

\usepackage[colorlinks=true,linkcolor=blue,citecolor=blue]{hyperref}
\usepackage[utf8]{inputenx}
\usepackage{pgfplots}
\usepackage{graphicx,wrapfig}
\usepackage{amsmath,amsfonts}
\usepackage{mathrsfs}
\usepackage[left=1in,top=1in,right=1in,bottom=1.3in]{geometry}
\usepackage{float}
\usepackage{mathtools}
\usepackage{listings}
\usepackage{color}
\usepackage[english]{babel}
\usepackage{amssymb}
\usepackage{amsthm}
\usepackage[ruled,vlined,linesnumbered]{algorithm2e}

\SetAlFnt{\small}
\SetAlCapFnt{\small}
\SetAlCapNameFnt{\small}
\SetAlCapHSkip{0pt}
\IncMargin{-\parindent} 
\usepackage{nicefrac}

\usepackage[square,numbers, sort&compress]{natbib}
\usepackage{sidecap}
\usepackage{subcaption}
\usepackage{fancyhdr}
\usepackage{multirow}
\usepackage[skip = 6pt]{parskip}
\usepackage{tagging}
\usepackage[T1]{fontenc}
\pgfplotsset{compat=1.18}

\usepackage{cleveref}
\usepackage{textcomp}

\usepackage{bm}
\usepackage{bbm}

\usepackage{array}

\allowdisplaybreaks

\usepackage{footmisc}
\usepackage[most]{tcolorbox}
\usepackage{tabularx}

\newtheoremstyle{newthm}
{5pt}
{}
{\itshape}
{}
{\bfseries}
{:}
{.5em}
{}

\theoremstyle{newthm}
	\newtheorem{theorem}{Theorem}[section]
    \newtheorem*{theorem*}{Theorem}
	\newtheorem{lemma}[theorem]{Lemma}
	\newtheorem{proposition}[theorem]{Proposition}
	\newtheorem{corollary}[theorem]{Corollary}
	
	\newtheorem{definition}[theorem]{Definition}
	\newtheorem{example}[theorem]{Example}
    \newtheorem*{Assumption}{Assumption}
\theoremstyle{remark}
	\newtheorem{remark}[theorem]{Remark}

\makeatletter
\def\namedlabel#1#2{\begingroup
   \def\@currentlabel{#2}%
   \label{#1}\endgroup
}
\makeatother

\newenvironment{assumption}[2]{\begin{Assumption}[\textbf{#1}] \namedlabel{#2}{(\textbf{#1})} }{\end{Assumption}}

\usepackage{tikz}
\usepackage{tikz-cd}
\usetikzlibrary{fit,matrix}

\usepackage[shortlabels]{enumitem}

\DeclareMathOperator*{\argmax}{arg\,max}

\newcommand{\agents}{\mathcal{N}}
\newcommand{\coalitions}{\mathcal{C}}
\newcommand{\edges}{\mathcal{E}}
\newcommand{\coalition}{C}
\newcommand{\networks}{X}
\newcommand{\network}{x}
\newcommand{\networkbis}{y}
\newcommand{\game}{\Gamma}
\newcommand{\payoff}{u}
\newcommand{\capacity}{\alpha}
\newcommand{\budget}{\beta}
\newcommand{\networksab}{X^{\capacity, \budget}}

\newcommand{\alg}{Generalized Deferred Acceptance}
\newcommand{\algshort}{\textbf{GDA}}

\title{A Unified Framework for Weighted Hypergraphic Networks and Fractional Matching}
\usepackage{authblk}

\author[*1]{Rémi Castera}
\author[*1]{Julien Fixary}
\author[1]{Rida Laraki}

\affil[1]{Moroccan Center for Game Theory, University Mohammed VI Polytechnic, Rabat, Morocco}
\affil[*]{{\small First authors, equal contribution \vspace{0.3cm}}}

\begin{document}
	
	\maketitle
\begin{abstract}
Network formation theory studies how agents create and maintain relationships, and the stability of those relationships with respect to individual incentives. A central stability concept in this literature is pairwise stability, introduced by Jackson and Wolinsky (1996) for unweighted networks (agents are either connected or not) and later extended by Bich and Morhaim (2020) to weighted networks (connections can have different intensities).
In this paper, we pursue two main objectives. First, we extend the notion of stability to networks defined on hypergraphs, where relationships may involve more than two agents simultaneously and where agents face budget constraints on the sum of the intensity of all their connections. We introduce a stability concept that preserves the core intuition of pairwise stability while generalizing it to relationships involving more than two agents, and that accounts for budget constraints. Second, we propose a stronger notion that we call full stability, inspired by stability concepts from matching theory, in which agents are allowed to adjust multiple connections simultaneously rather than through single-link deviations.
We give existence results for both stability notions under various assumptions, as well as explicit solutions or algorithms, and provide counter-examples for most cases that do not satisfy those assumptions, establishing an almost complete theory. Our framework provides a unified approach to constrained network formation in hypergraphic settings and builds a conceptual bridge between the theories of weighted network formation and fractional matching.
\end{abstract}
    
	\newpage

	\section{Introduction}

    The aim of this paper is to establish a bridge between \emph{strategic weighted network formation theory} and \emph{fractional matching theory}.

    On the one hand, following the seminal contribution of \citet{jackson_wolinsky_1996}, network formation theory has been extensively studied over the last decades from a theoretical perspective \cite{demange_2005,jackson_watts_2001,JACKSON2002,JACKSON2002265,DUTTA1997322,JACKSON2005,calvo_2009,HERINGS2009526,KIRCHSTEIGER201697,BICH2023105577,herings_2024,bich_morhaim_2020,fixary_2022,Fixary_2025}.
    Its purpose is to analyze how networks emerge endogenously from the strategic interactions of self-interested agents, and how individual incentives shape the architecture of social, economic, or informational networks.
    In parallel, a large applied literature has used network formation models to study, among others, juvenile delinquency and conformism \cite{patacchini}, social networks in education \cite{calvo_patacchini_zenou}, R\&D collaboration networks \cite{konig_2019}, and social interactions within cities \cite{HELSLEY2014426}. Within this literature, particular attention has been devoted to stability concepts that capture the robustness of a network to individual or pairwise deviations.   
    One of the most prominent among them is \emph{pairwise stability}, introduced by \citet{jackson_wolinsky_1996} for unweighted networks where agents' decisions to form a link are binary. It was later extended by Bich and Morhaim \cite{bich_morhaim_2020} to weighted networks, in which agents are still represented as nodes of a graph, and the intensity of each link is determined endogenously and strategically by the two agents forming it. The link intensity represents for instance, the amount of time, effort, or resources devoted to their interaction. Given exogenous preferences over all possible networks, a weighted network is \emph{pairwise stable} if no agent wishes to unilaterally reduce the intensity of any of their existing links, and if no pair of agents can jointly increase the intensity of their mutual link in a way that benefits both of them.
    The weighted framework significantly broadens the scope of applications.
    In financial networks, weighted links naturally represent the exposure size in interbank lending.
    In communication networks, weighted links may represent attention.
    In labor markets, weights capture the intensity of referrals or mentoring relationships.
    In production, and R\&D networks, weighted links may correspond to trade volumes or collaboration effort.

    On the other hand, matching theory, which has a much longer history, studies how agents form disjoint pairs when they all have \emph{ordinal preferences}, i.e., when they rank all potential partners from best to worst. Using a cardinal payoff interpretation, a matching model can be formally be interpreted as a network formation game with linear payoff functions without externalities: the ranking is given by sorting the coefficients associated with each potential partner in the linear function, and there are no externalities since agents are only affected by their own connections. The fact that a matching consists of disjoint pairs implies that agents have a budget, since they can have only one partner, which is a noticeable difference with network models.
    The definition of stability, as well as the celebrated \emph{Deferred Acceptance} algorithm that finds a stable matching in any instance of the problem, were introduced by Gale and Shapley \cite{gale_college_1962} for the bipartite and discrete case. Since a stable matching is not always guaranteed to exist in non-bipartite graphs (called the \emph{roommate problem}), the literature has mostly focused on the bipartite case, mainly to study the structure of the set of stable matchings in the original model as well as dozens of variants \cite{mcVittie_marriage_1971, knuth_mariages_1976, dubins_machiavelli_1981, roth_economics_1982, roth_interns_1984, roth_allocation_1986, Balinski_97, Baiou_00b}. Later, researchers started investigating stability in \emph{fractional matchings}, i.e., matching where agents can divide their time between several partners. This extension is quite natural, as it corresponds to many real-world situations, for instance, doctors that divide their time between several hospitals. This study was initiated by \citet{Vandevate_89}, and quickly gained attention from the community \cite{Rothblum_92, Roth_93, Baiou_00, Baiou_02, Alkan_03, Moulin_04}. An interesting property of fractional matchings is that they recover existence of stable matchings in non-bipartite graphs, that have then gained more attention \cite{abeledo_94, teo_98, Aharoni_2003}. Generalizing even further, some researchers have studied stable matching in hypergraphs, i.e., when agents do not form just pairs but coalitions of any size and potentially overlapping \cite{Aharoni_2003, rostek_20, csaji_21, biro_25}.

    While there are many similarities between network formation theory and matching theory, the two literatures have evolved independently for the most part. Very recently, \citet{sadler_23} proposed a unified approach to unweighted network formation and discrete matching, assuming that payoffs have no externalities. In this article, our goal is to propose a framework that unifies weighted network formation and fractional matching in the most general possible setting, notably, we consider hypergraphs rather than just graphs and we allow for externalities.

    \subsection{Contributions}

    Strategic network formation theory guarantees the existence of pairwise stable networks under broad assumptions on preferences -- such as continuity and quasiconcavity -- but typically abstracts from budget constraints on agents' ability to allocate time, effort, or resources across relationships (see, e.g., \citet{bich_morhaim_2017}).
    By contrast, fractional matching and two-sided many-to-many matching models establish a stronger form of stability under budget constraints, albeit for highly restrictive classes of payoff functions, i.e., linear and free of externalities\footnote{Matching models generally use ordinal preferences, but these can be modeled as linear payoff functions without externalities.} (see, e.g., \citet{Baiou_02,Aharoni_2003}).
    We bridge these approaches by introducing a unified framework in which coalitions form links, agents face budget constraints, and payoff functions may be non-separable, non-linear, non-monotonic, and subject to externalities.

    More specifically, the first contribution of this paper is to extend the notion of pairwise stability for weighted networks in two distinct directions.
    First, instead of restricting attention to bilateral connections, we allow for more general coalitions to be assigned a weight.
    The underlying intuition in a time-sharing interpretation, is that agents may allocate their time not only to pairwise interactions, but also to activities involving larger groups.
    Second, we explicitly incorporate feasibility constraints into the model, namely, coalitional capacity constraints and individual budget constraints.
    Within the same time-sharing interpretation, these assumptions reflect two natural limitations: any coalition is constrained in the total amount of time its members can spend together, and each agent has a limited amount of time that must be allocated across all coalitions in which they participate.
    To account simultaneously for coalitional interactions and feasibility constraints, we introduce a stability concept -- referred to as \emph{stability} in this paper -- that extends both the notion of pairwise stability of \citet{jackson_wolinsky_1996}, and the one of \citet{bich_morhaim_2020}.
    Under similar continuity and quasiconcavity assumptions on the payoff functions as in \cite{bich_morhaim_2020}, we show that a stable weighted network always exists (see \Cref{thm:hypergraph_stable}).

    Nevertheless, the incorporation of budget constraints naturally raises the question of whether the stability condition described above is sufficiently demanding in this enriched framework.
    In particular, it seems natural to allow agents, when considering a deviation, to reallocate the weights of the coalitions in which they are involved, rather than only allowing them to change the weight of a single coalition at a time.
    In response to this observation, the second contribution of this paper is to introduce \emph{full stability} for strategic network games, a stronger notion of the classical pairwise stability that explicitly allows agents to reallocate their incident weights when considering deviations by a coalition.
    This notion remains closely related to the spirit of pairwise stability: it captures the absence of profitable unilateral decreases, while combining the idea of a consensual increase in a coalition's weight with the possibility of unilateral decreases on other incident coalitions in a natural and economically meaningful way.
    This is precisely where the link with matching theory is made. Indeed, while matching models (roommate, two-sided, many-to-many, etc.) can formally be seen as particular cases of network formation, they usually impose budget constraints (for instance, a university cannot recruit more than a certain number of students or professors), and they rely on a stronger notion of pairwise stability: an unmatched pair can -- and subject to their budget constraints, should -- break existing links in order to form a new pair that benefits to both parties. We show that matching problems can be formulated as network problems in our framework, and that full stability is equivalent to the classical notion of stability in matching theory (except for hypergraphs where the two notions differ very slightly). In particular, all existence results for stable fractional matchings in bipartite graphs become special cases of \Cref{thm:bipartite_full}. Moreover, we show in \Cref{prop:discrete_matching} that our framework also encapsulates discrete matching problems.

    However, the existence of a fully stable network is far from obvious, and such networks generally do not exist without additional assumptions on payoffs or constraints.
    We provide two simple counter-examples illustrating the non-existence of fully stable networks. The first considers a linear payoff function without budget constraints and with externalities (see \Cref{thm:counter}), while the second focuses on a quasiconcave but non concave payoff (convex, monotone, and separable), with budget constraints and without externalities (see \Cref{thm:counter_2}).
    The former highlights a major difficulty linked to the presence of negative externalities, whereas the latter illustrates how the lack of concavity can undermine existence even in the absence of external effects.
    Both examples exhibit different forms of cyclic improvement patterns in agents’ payoffs.
    To address these difficulties, we pursue two complementary approaches.
    The first focuses on the cycle identified in \Cref{thm:counter} in the unconstrained case. We provide a simple condition under which stability and full stability coincide (see \Cref{thm:hypergraph_full_wo_cnst}). We also study a particular case satisfying this condition, for which the explicit computation of all fully stable networks is possible (see \Cref{prop:min_argmax}).
    The second approach considers bipartite graphs (e.g., two-sided many-to-many fractional matchings) with budget constraints and positive externalities, thereby addressing the issues raised by \Cref{thm:counter_2}. In this setting, we provide an explicit algorithm that guarantees the existence of a fully stable network when payoffs are concave (see \Cref{thm:bipartite_full}). Interestingly, both the algorithm and the existence result remain valid for discrete many-to-many matching problems with payoffs that only need to be continuous, separable and concave rather than linear, which is, to our knowledge, a new contribution (see \Cref{prop:discrete_matching}).

    \Cref{fig:results_table} provides an overview of the paper's main contributions.
    \Cref{thm:hypergraph_stable} establishes the existence of a stable network in hypergraphic network games with budget constraints, under the assumption that payoff functions are continuous and quasiconcave in the intensity of each coalition separately\footnote{This requirement is substantially weaker than global quasiconcavity.}.
    The proof relies on the existence of a maximal element of an appropriate preference relation, following the approach of \citet{bich_morhaim_2017}.
    The remaining results concern full stability.
    \Cref{thm:counter} shows that negative externalities rule out the existence of fully stable networks, even with linear payoff functions and bipartite graphs.
    \Cref{thm:counter_2} further demonstrates that, even without externalities, full stability may fail under budget constraints when payoff functions are not concave -- in fact, the counter-example features convex, increasing, and separable payoff functions.
    In the absence of budget constraints, existence is recovered under separability (see \Cref{prop:min_argmax}).
    Under budget constraints, existence is restored under concavity in bipartite graphs, both in fractional and integer settings (see \Cref{thm:bipartite_full}), using a generalized propose-dispose algorithm.

    \begin{figure}[ht]
        \centering
        \hspace{-0.4cm}
        \makebox[\textwidth]{
        \begin{tikzpicture}
        \node[scale=0.805, xshift = 0 cm] (table){
        \begin{tabular}{|l|c|c|c|c|}
        \hline
        & \multicolumn{2}{c|}{stable} & \multicolumn{2}{c|}{fully stable} \\
        \cline{2-5}
        & no constraints & constraints & no constraints & constraints \\
        \hline
        bipartite graph  & Bich\&Morhaim \cite{bich_morhaim_2020} & \textbf{\Cref{thm:hypergraph_stable}} & \begin{tabular}{@{}c@{}}\textbf{\Cref{thm:counter}} (counter-ex.) (L) \\ \textbf{\Cref{prop:min_argmax}} (S, PE)\end{tabular} & \begin{tabular}{@{}c@{}}Baïou\&Balinski \cite{Baiou_02} (L, NE) \\ \textbf{\Cref{thm:counter}} (counter-ex.) (L) \\\textbf{\Cref{thm:bipartite_full}} (S, C, PE)\end{tabular}\\
        \hline
        graph  \begin{tabular}{@{}c@{}} \phantom{a} \\ \phantom{a}\end{tabular}         & Bich\&Morhaim \cite{bich_morhaim_2020} & \textbf{\Cref{thm:hypergraph_stable}} & \begin{tabular}{@{}c@{}}\textbf{\Cref{thm:counter}} (counter-ex.) (L) \\\textbf{\Cref{prop:min_argmax}} (S, PE)\end{tabular} & \begin{tabular}{@{}c@{}}Aharoni\&Fleiner \cite{Aharoni_2003} (L, NE)  \\\textbf{\Cref{thm:counter}} (counter-ex.) (L) \\\textbf{\Cref{thm:counter_2}} (counter-ex.) (S, NE) \end{tabular} \\
        \hline
        hypergraph  \begin{tabular}{@{}c@{}} \phantom{a} \\ \phantom{a}\end{tabular}    & \textbf{\Cref{thm:hypergraph_stable}} & \textbf{\Cref{thm:hypergraph_stable}} & \begin{tabular}{@{}c@{}}\textbf{\Cref{thm:counter}} (counter-ex.) (L) \\ \textbf{\Cref{prop:min_argmax}} (S, PE)\end{tabular} & \begin{tabular}{@{}c@{}}Aharoni\&Fleiner \cite{Aharoni_2003} (L, NE)  \\\textbf{\Cref{thm:counter}} (counter-ex.) (L) \\\textbf{\Cref{thm:counter_2}} (counter-ex.) (S, NE) \end{tabular}\\
        \hline
        \end{tabular}
        };

        \node[xshift = -8.8 cm, yshift=0.7cm] (bip) {};
        \node[xshift = -8.8 cm, yshift=-0.4cm] (gr) {};
        \node[xshift = -8.8 cm, yshift=-1.5cm] (hyp) {};
        \node[xshift = -4.9 cm, yshift=2.1cm] (snc) {};
        \node[xshift = -2.5 cm, yshift=2.1cm] (sc) {};
        \node[xshift = 1.2 cm, yshift=2.1cm] (fsnc) {};
        \node[xshift = 6.3 cm, yshift=2.1cm] (fsc) {};

        \draw[->, thick] (hyp.north west) to [out=150,in=-150] (gr.south west);
        \draw[->, thick] (gr.north west) to [out=150,in=-150] (bip.south west);
        \draw[->, thick] (sc.north west) to [out=130,in=50] (snc.north east);
        \draw[->, thick] (fsc.north west) to [out=130,in=50] (fsnc.north east);
        \draw[->, thick] (fsc.north west) to [out=130,in=50] (sc.north east);
        \draw[->, thick] (fsnc.north west) to [out=130,in=50] (snc.north east);
        \end{tikzpicture}}
        \caption{Summary of the existing work and of our contributions. Every entry indicates a proof of existence of a stable (or fully stable) network, in the context (type of graph, constraints or not) defined by the row and column labels, except for Theorems \ref{thm:counter} and \ref{thm:counter_2} which are counter-examples. The arrows indicate the nestedness of concepts: for instance, an existence result for hypergraphs in a given column implies existence in graphs, and thus in bipartite graphs, in the same column. For counter-examples, the implications go in the opposite direction. Some results are only true under some assumptions on the payoff functions, indicated in parenthesis: L = linear, S = additively separable, C = concave, NE = no externalities, PE = positive externalities. All results require at least continuous and quasiconcave payoff functions; those assumptions are thus omitted from the table.}
        \label{fig:results_table}
    \end{figure}

    \subsection{Literature on Network Formation Theory and Pairwise Stability}

    Since the seminal contribution of Jackson and Wolinsky \cite{jackson_wolinsky_1996}, the concept of \emph{pairwise stability} has become one of the central solution concepts in network formation theory. 

    In subsequent work, Jackson and Watts \cite{jackson_watts_2001} investigate the existence of pairwise stable unweighted networks and show that it is not automatic, since improving paths may cycle.
    They analyze such improving paths and identify conditions under which pairwise stable unweighted networks exist and can be reached through local adjustments.

    In his survey in \cite{demange_2005}, Jackson provides a systematic overview of models of network formation and stability notions.
    He reviews the definition of pairwise stability introduced by Jackson and Wolinsky, discusses its behavioral interpretation, and compares it with alternative notions of stability and equilibrium used in the literature.
    The survey also examines the relationship between pairwise stability, efficiency, and welfare, and illustrates how different modeling assumptions affect the set of stable unweighted networks across a variety of applications.

    More recent work has moved beyond binary links to study weighted networks, in which links have continuous intensities, and pairwise stability in this context.
    This extension is motivated by applications in which relationships vary in strength, capacity, or frequency.
    Bich and Morhaim \cite{bich_morhaim_2020} establish existence results for pairwise stable weighted networks under general continuity and quasiconcavity conditions, using fixed-point arguments.

    Beyond existence, a growing literature investigates the global and generic properties of the set of pairwise stable weighted networks.
    Bich and Fixary \cite{fixary_2022} study the topological structure of the graphs of pairwise stable weighted networks in network formation games.
   
    They show that, for broad classes of network games, generically the number of pairwise stable weighted networks is odd.
    This result implies robust existence and reveals deep connections between pairwise stability and classical equilibrium concepts in game theory and general equilibrium theory.
    Building on this perspective, Fixary \cite{Fixary_2025} further investigates the global geometry of graphs of pairwise stable weighted networks, establishing, in particular configurations, their unknottedness and related topological properties that shed light on network dynamics and equilibrium selection.

    Finally, the computational aspect of pairwise stability has gained increasing attention.
    Herings and Zhan \cite{herings_2024} develop a homotopy-based approach to compute pairwise stable weighted networks by reformulating pairwise stability as a Nash equilibrium of an auxiliary link-based game and applying a linear tracing procedure.
    Complementing this contribution, Chen, Tao, and Zhan \cite{chen_tao_zhan} focus on the efficient computation and selection of pairwise stable weighted networks in network formation games under differentiability and concavity assumptions. They introduce a logarithmic tracing procedure and a path-following algorithm, and demonstrate through numerical experiments that their method outperforms standard linear tracing and exhaustive search approaches.

    \subsection{Literature on Matching Theory}

    Matching theory started as a combinatorial optimization problem, where there were no preferences but the goal was to study the geometry of the set of feasible matchings, defined as a set of disjoint edges in a graph, and to find the largest one \cite{konig_graphs_1931, hall_representatives_1935, Edmonds1965, Lovasz_86}. The introduction of preferences and of the concept of stability is due to \citet{gale_college_1962}, along with the Deferred Acceptance algorithm that computes a stable matching for any instance of the problem. Their model is a many-to-one matching problem, i.e., the graph is bipartite, agents on the left side can only be matched to one agent, while agents on the right side can be matched to multiple agents, up to some budget constraint. The model was then extended to many-to-many matching (both sides have budgets greater than one) \cite{roth_many_84}. In terms of existence results, in the discrete case, the theory has been quite straightforward ever since the seminal paper \cite{gale_college_1962}: in a bipartite graph, a stable matching always exists, while in a non-bipartite graph (called the roommate problem) stable matchings may fail to exist.

    In a similar fashion, fractional matching was first studied as an object of optimization in graph theory and linear programming \cite{berge_79, lovasz_79, Furedi_81} before researchers started to inquire about their stability in the presence of preferences, basing their definition of stability on the one existing in the discrete case. \citet{Vandevate_89} proposed the first description of the polytope of fractional one-to-one stable matchings, later refined and simplified \cite{Rothblum_92, Roth_93}. Baïou and Balinski extended the model to many-to-one matching \cite{Baiou_00} and then to many-to-many matching \cite{Baiou_02}, proving existence of stable fractional matchings in this setting and providing two algorithms to compute them. The fractional roommate problem (i.e., matching in non-bipartite graphs) has also been studied, with \citet{abeledo_94} proving the existence of fractional matchings using a reduction to the bipartite case, and \citet{teo_98} who studies the polyhedral structure of the problem. Finally, \citet{Aharoni_2003} extend the model to hypergraphs, that is, when agents form coalitions of any size rather than just pairs, and prove the existence of a stable matching in this context. More recently, \citet{Caragiannis_19} studies a model of matching with cardinal utilities with the additional objective of maximizing social welfare, showing that fractional matchings can achieve better outcomes than integral ones.

    The two works closest to ours are \citet{Baiou_02}, and \citet{Aharoni_2003}. In the bipartite case, our results are more general since they allow for any separable and concave payoffs with positive externalities, rather than ordinal preferences, which can be represented equivalently by linear payoffs without externalities. For non-bipartite graphs and hypergraphs, \citet{Aharoni_2003} remains the state of the art, in the sense that we cannot prove existence of fully stable networks in constrained problems under weaker assumptions than theirs. However, in "weaker" contexts, like full stability in unconstrained problems, or stability in the network sense, which is weaker than the matching notion of stability, our existence results require fewer assumptions than those that are implied by Aharoni and Fleiner's result.

    \subsection{Outline}
    
    The paper is organized as follows. \Cref{sec:remind_net} introduces the basics definitions and notation of network formation theory, along with the classical existence result. \Cref{sec:hyp_net} extends the model to hypergraphs and provides a proof of existence in this context. In \Cref{sec:full_def}, we introduce the notion of full stability, explain how it links network formation theory to matching theory, and provide counter-examples showing that a fully stable network does not always exist under the assumptions of the previous existence results. Then, \Cref{sec:full_exist} shows that existence can be recovered under several different sets of assumptions. Finally, we discuss the remaining open questions in \Cref{sec:disc}.

	\section{Reminders from Network Formation Theory} \label{sec:remind_net}
	
	In this section, we recall the main concepts of strategic network formation theory (the notion of a network, pairwise stability, and related ideas).
	
	First of all, a \emph{set of agents} is a nonempty finite set $\agents$ and the \emph{set of coalitions (with respect to $\agents$)} is defined as the set
	$$\coalitions = 2^\agents \backslash \{\emptyset\}.$$
	Given a set $\agents$ of agents, a \emph{set of edges} is a nonempty subset $\edges$ of $\coalitions$.
    A couple $(\agents,\edges)$ is called an \emph{hypergraph}, unless $\vert \coalition \vert = 2$, for every $\coalition \in \edges$, then it is simply called a \emph{graph}. For any set of edges $\edges$,
    the \emph{set of (weighted) networks (with respect to $\agents$ and $\edges$)} is defined as
	$$\networks = [0,1]^\edges.$$
	
	Before continuing with our preliminaries, we introduce some practical notation that will be used throughout the paper.
	Let $\agents$ be a set of agents, $\edges$ be a set of edges, and let $\coalition \in \edges$ be a coalition.
	Define
	$$\edges_{-\coalition} = \edges \backslash \{\coalition\} \text{ and } \networks_{-\coalition} = [0,1]^{\edges_{-\coalition}}.$$
	For any $\networkbis_\coalition \in [0,1]$, and any $\network_{-\coalition} \in \networks_{-\coalition}$, $(\networkbis_\coalition,\network_{-\coalition}) \in \networks$ is the network defined by $(\networkbis_\coalition,\network_{-\coalition})_{\coalition'} = (\network_{-\coalition})_{\coalition'}$, for any $\coalition' \in \edges_{-\coalition}$, and $(\networkbis_\coalition,\network_{-\coalition})_\coalition = \networkbis_\coalition$.
	Last, for any network $\network \in \networks$, $\network_{-\coalition} \in \networks_{-\coalition}$ corresponds to the tuple of weights in $\network$ where the weight associated to $\coalition$ is missing.

	We now recall the definition of one of the main structure in strategic network formation theory.

    \begin{definition}[Network game]
		A \emph{network game} is a tuple
		$$\game = (\agents,\edges,\payoff),$$
		where $\agents$ is a set of agents, $\edges$ is a set of edges, and $\payoff: \networks \to \mathbb{R}^\agents$.
		For every agent $i \in \agents$, the map $\payoff_i$ which associates to any network $\network \in \networks$ the real number $\payoff_i(\network)$ is called the \emph{payoff function of agent $i$}.
	\end{definition}

    We also recall the pairwise stability concept of Bich and Morhaim \cite{bich_morhaim_2020}, for network games.

	\begin{definition}[Pairwise stable network]
		Let $\game = (\agents,\edges,\payoff)$ be a network game such that $(\agents,\edges)$ is a graph.
		A network $\network \in \networks$ is \emph{pairwise stable with respect to $\game$} if for every coalition $\{i,j\} \in \edges$ and every weight $\networkbis_{\{i,j\}} \in [0,1]$, the following two conditions hold:
		\begin{enumerate}
			\item If $\networkbis_{\{i,j\}} \leq \network_{\{i,j\}}$, then, for every $k \in \{i,j\}$, $\payoff_k(\networkbis_{\{i,j\}},\network_{-\{i,j\}}) \leq \payoff_k(\network)$.
			\item If $\network_{\{i,j\}} \leq \networkbis_{\{i,j\}}$, then there exists $k \in \{i,j\}$, $\payoff_k(\networkbis_{\{i,j\}},\network_{-\{i,j\}}) \leq \payoff_k(\network)$.
		\end{enumerate}
	\end{definition}

    Compared to the notion of Nash equilibrium in game theory, pairwise stability combines both noncooperative and cooperative features. On the one hand, there is no unilateral incentive to cut an existing link; on the other hand, the formation of a new link requires mutual consent, reflecting the absence of profitable bilateral deviations.

    Bich and Morhaim \cite{bich_morhaim_2020} also provide an existence result for their notion, for which two minimal assumptions are required.

    \begin{assumption}{C0}{asspt:cont}
        A network game is  \emph{continuous} if each agent $i$'s payoff function $\payoff_i$ is continuous.
    \end{assumption}
    
    \begin{assumption}{QC}{asspt:quasiconc}
        A network game is  \emph{quasiconcave} if each agent $i$'s payoff function $\payoff_i$ is component-wise quasiconcave\footnote{Formally, this means that $\forall i \in \agents, \forall \coalition \in \edges, \forall \network_{-\coalition}\in\networks_{-\coalition}$, $\payoff_i(\cdot, \network_{-\coalition})$ is quasiconcave. This assumption is weaker that global quasiconcavity.}.
    \end{assumption}
    
	\begin{theorem*}[Bich--Morhaim \cite{bich_morhaim_2020}]
	   Let $\game = (\agents,\edges,\payoff)$ be a continuous \ref{asspt:cont} and quasiconcave \ref{asspt:quasiconc} network game such that $(\agents,\edges)$ is a graph. Then $\game$ admits a pairwise stable network.
	\end{theorem*}

    We say that the assumptions \ref{asspt:cont} and \ref{asspt:quasiconc} are minimal because removing either one leads to network games that admit no stable network, as also shown by Bich and Morhaim \cite{bich_morhaim_2020}.

   There are several limitations to the Bich--Morhaim model.
   In many network formation settings, link intensities naturally represent scarce resources that agents must allocate across relationships. Beyond time allocation, budget constraints arise naturally in many strategic network formation environments. In R\&D and innovation networks, links represent joint research projects, and agents face limited research budgets or personnel that must be allocated across collaborations. In financial networks, links may capture exposure through joint investments, risk-sharing agreements, or credit relationships, where capital constraints limit the total intensity of bilateral or multilateral ties. In production and supply-chain networks, firms form partnerships or coalitions to share inputs or technologies, subject to capacity and cost constraints. More generally, coalitions naturally arise when productivity depends on the joint participation of several agents, as in research consortia, syndicate loans, startup founding teams, or collaborative platforms. In such settings, agents must decide not only which coalitions to join, but also how to allocate scarce resources -- capital, effort, or capacity -- across multiple overlapping group interactions.

	\section{Stability in Hypergraphic Networks with Constraints} \label{sec:hyp_net}
	
	Motivated by the above considerations, we extend in this section the classical framework of network formation along two dimensions.
	First, interactions are no longer restricted to bilateral links: weights are assigned to coalitions of agents, allowing us to model multilateral relationships. Second, agents' choices are subject to constraints, reflecting limitations on the total weight they can allocate across coalitions, as well as upper bounds on the weight that each coalition may receive.
	
	Within this framework, we introduce a notion of stability that naturally extends pairwise stability to weighted hypergraphic networks.
	The core intuition of pairwise stability is preserved: a unilateral decrease in a coalition's weight is feasible whenever at least one member benefits from it, whereas a joint increase can occur only if all members agree to it and the resulting allocation remains compatible with their individual constraints.
	This approach allows us to explicitly account for constrained agents in the analysis of network formation.
	
	For any set $\agents$ of agents, any set $\edges$ of edges, any family $\capacity = (\capacity_\coalition)_{\coalition \in \edges}$ of real numbers in $[0,1]$ (called a \emph{capacity vector})\footnote{The bound is arbitrarily chosen to be 1 without loss of generality, it could be changed to any positive real number without affecting any of the results presented in this paper.}, and any family $\budget = (\budget_i)_{i \in \agents}$ of positive real numbers (called a \emph{budget vector}), define
	$$\networksab = \{\network \in \prod_{\coalition \in \edges} [0,\capacity_\coalition] : \forall i \in \agents, \sum_{\coalition \ni i} \network_\coalition \leq \budget_i\}.$$

    \begin{definition}[Constrained network game]
        A \emph{constrained network game} is a tuple
        $$\game = (\agents,\edges,\payoff,\capacity,\budget),$$
		where $(\agents,\edges,\payoff)$ is a network game, $\capacity$ is a capacity vector, and $\beta$ is a budget vector.
    \end{definition}

	Notice that an (unconstrained) network game is a special case of a constrained network game with each $\alpha_\coalition$ being equal to one, and each $\beta_i$ equals to $+ \infty$. From now on, the phrase "network game" will refer to general constrained network games unless otherwise specified.
    
    We now introduce our notion of stability, adapted to a framework in which each coalition is assigned a weight and agents' choices are subject to constraints.
	
	\begin{definition}[Stable network]
		Let $\game = (\agents,\edges,\payoff, \capacity, \budget)$ be a network game.
		A network $\network \in \networks$ is \emph{stable with respect to $\game$} if for every coalition $\coalition \in \edges$ and every weight $\networkbis_\coalition \in [0,\capacity_\coalition]$, the following three conditions hold:
		\begin{enumerate}
			\item $\network \in \networksab$.
			\item If $\networkbis_\coalition \leq \network_\coalition$, then, for every $i \in \coalition$, $\payoff_i(\networkbis_\coalition,\network_{-\coalition}) \leq \payoff_i(\network)$.
			\item If $\network_\coalition \leq \networkbis_\coalition$, and if, for every $i \in \coalition$, $\payoff_i(\networkbis_\coalition,\network_{-\coalition}) > \payoff_i(\network)$, then there exists $j \in \coalition$ such that
			$$\sum_{\substack{\coalition' \neq \coalition \\\coalition' \ni j}} \network_{\coalition'} + \networkbis_\coalition > \budget_j.$$
		\end{enumerate}
	\end{definition}
	
	This notion is in the spirit of the original concept of pairwise stability introduced by Bich and Morhaim.
	However, the bilateral improvement condition is modified as follows: whenever the members of a coalition would benefit from increasing the coalition’s level of connection, there must exist at least one agent whose constraint would be violated by such an increase.
	
	We now present an existence theorem for stable networks under standard continuity and quasiconcavity assumptions.
	
	\begin{theorem} \label{thm:hypergraph_stable}
		Let $\game = (\agents,\edges,\payoff, \capacity, \budget)$ be a continuous \ref{asspt:cont} and quasiconcave \ref{asspt:quasiconc} network game. Then $\game$ admits a stable network.
	\end{theorem}

    \begin{proof}[Sketch of proof.]
    The proof follows the spirit of Theorem 6.1 in Bich--Morhaim \cite{bich_morhaim_2017}.
    Namely, the idea is to apply a result due to Yannelis and Prabhakar -- see Theorem 5.2 in \cite{yannelis_prabhakar_1983} -- to the correspondence $\Psi: \networksab \to 2^{\networksab}$ defined as follows.
	For any $\network,\network' \in \networksab$, we have $\network' \in \Psi(\network)$ if and only if there exists a coalition $\coalition \in \edges$ such that one of the following conditions holds:
	\begin{enumerate}
		\item There exists a weight $\networkbis_\coalition \in [0,\network_\coalition[$ such that $\network' = (\networkbis_\coalition,\network_{-\coalition})$ and
		$$\exists i \in \coalition, \payoff_i(\network') > \payoff_i(\network).$$
		\item There exists a weight $\networkbis_\coalition \in ]\network_\coalition,\capacity_\coalition]$ such that $\network' = (\networkbis_\coalition,\network_{-\coalition})$ and
        $$[\forall i \in \coalition, \payoff_i(\network') > \payoff_i(\network)] \wedge [\forall i \in \coalition, \sum_{\substack{\coalition' \neq \coalition\\\coalition' \ni i}} \network_{\coalition'} + \networkbis_\coalition \leq \budget_i].$$
	\end{enumerate}
	The full proof is provided in Appendix \ref{app.proof:hypergraph_stable}.
    \end{proof}

	\section{Full Stability: Definition and Counter-examples} \label{sec:full_def}

    Now that constraints are imposed, it is however no longer obvious that pairwise stability as currently defined remains the most appropriate concept to capture network stability. To see this, consider the following payoff functions:
    $$\begin{aligned}
        \payoff_1(\network) &= \network_{\{1,2\}} + 2 \network_{\{1,3\}}, \\
        \payoff_2(\network) &= 2 \network_{\{1,2\}} + \network_{\{2,3\}}, \\
        \payoff_3(\network) &= 2 \network_{\{1,3\}} + \network_{\{2,3\}},
    \end{aligned}$$
    together with identical budget constraints $\budget_1 = \budget_2 = \budget_3 = 1$.
    In this setting, the network $x = (\network_{\{1, 2\}}, \network_{\{1, 3\}}, \network_{\{2, 3\}}) = (1,0,0)$ is stable.
    Similarly, the network $\network' = (0,1,0)$ is also stable.
    Yet these two networks lead to very different outcomes: the payoff vector associated with $\network$ is $(1,2,0)$, whereas the payoff vector associated with $\network'$ is $(2,0,2)$.
    Agents 1 and 3 are strictly better off under $\network'$ than under $\network$, but they cannot transition from $\network$ to $\network'$ through the deviations allowed by classical pairwise stability. This example is illustrated in \Cref{fig:ex_stable}.

    \begin{figure}[ht]
 \centering
\begin{tikzpicture}
\begin{scope}[every node/.style={thick}]
    \node (u1) at (0,-0.5) {$u_1 = 1$};
    \node (u2) at (2,-0.5) {$u_2 = 2$};
    \node (u3) at (1,2) {$u_3 = 0$};
    \node (x12) at (1,-0.2) {$1$};
    \node (x23) at (1.7,0.75) {$0$};
    \node (x13) at (0.3,0.75) {$0$};
\end{scope}

\begin{scope}[every node/.style={circle,thick,draw}]
    \node (1) at (0,0) {$1$};
    \node (2) at (2,0) {$2$};
    \node (3) at (1,1.5) {$3$};
\end{scope}

\begin{scope}[>={Stealth[black]},
              every node/.style={fill=white,circle},
              every edge/.style={draw=black, semithick}]
    \path [-] (1) edge (2);
    \path [-] (1) edge (3);
    \path [-] (3) edge (2);
\end{scope}
\end{tikzpicture}\hspace{2 cm}
\begin{tikzpicture}
\begin{scope}[every node/.style={thick}]
    \node (u1) at (0,-0.5) {$u_1 = 2$};
    \node (u2) at (2,-0.5) {$u_2 = 0$};
    \node (u3) at (1,2) {$u_3 = 2$};
    \node (x12) at (1,-0.2) {$0$};
    \node (x23) at (1.7,0.75) {$0$};
    \node (x13) at (0.3,0.75) {$1$};
\end{scope}

\begin{scope}[every node/.style={circle,thick,draw}]
    \node (1) at (0,0) {$1$};
    \node (2) at (2,0) {$2$};
    \node (3) at (1,1.5) {$3$};
\end{scope}

\begin{scope}[>={Stealth[black]},
              every node/.style={fill=white,circle},
              every edge/.style={draw=black, semithick}]
    \path [-] (1) edge (2);
    \path [-] (1) edge (3);
    \path [-] (3) edge (2);
\end{scope}
\end{tikzpicture}
\caption{Illustration of the limitation of the notion of stability in the presence of constraints. Left: network $\network = (1, 0, 0)$, right: network $\network' = (0, 1, 0)$. Both agents $1$ and $3$ are better off under $\network'$,  nonetheless $\network$ is stable according to the classical definition of pairwise stability.}
\label{fig:ex_stable}
\end{figure}
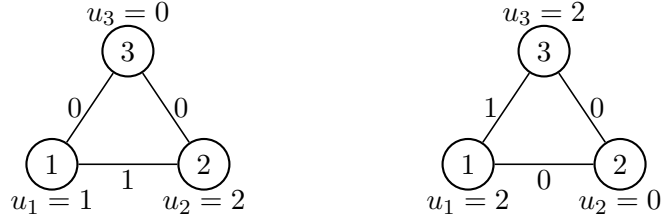
    
    To address this limitation, we introduce a stronger concept of stability called \emph{full stability}.
	This notion refines the previous one by allowing agents, when jointly increasing the weight of a coalition, to simultaneously decrease the weights of other coalitions to which they belong. This flexibility enables agents to reallocate weights in order to satisfy their individual constraints. As will be explained later, this notion also shares similarities with stability concepts in matching theory, thereby building a bridge between the two literatures.
    
	We now introduce some notation that will be used throughout the paper.
	Let $\agents$ be a set of agents, $\edges$ be a set of edges:
	\begin{itemize}
		\item Let $i \in \agents$ be an agent.
		Define
		$$\edges^i = \{\coalition \in \edges : \coalition \ni i\}, \ \edges^{-i} = \edges \backslash \edges^i, \ \networks^i = [0,1]^{\edges^i} \text{ and } \networks^{-i} = [0,1]^{\edges^{-i}}.$$
		For any $\networkbis^i \in \networks^i$, and any $\network^{-i} \in \networks^{-i}$, $(\networkbis^i,\network^{-i}) \in \networks$ is the network defined by $(\networkbis^i,\network^{-i})_\coalition = (\network^{-i})_\coalition$, for any $\coalition \in \edges^{-i}$, and $(\networkbis^i,\network^{-i})_\coalition = \networkbis^i_\coalition$, for any $\coalition \in \edges^i$.
		Last, for any $\network \in \networks$, $\network^{-i} \in \networks^{-i}$ corresponds to the tuple of weights in $\network$ where the weights associated to $i$ are missing.
		\item Let $\coalition \in \edges$ be a coalition.
		Define
		$$\edges^\coalition = \bigcup_{i \in \coalition} \edges^i, \ \edges^{-\coalition} = \edges \backslash \edges^\coalition, \ \networks^\coalition = [0,1]^{\edges^\coalition} \text{ and } \networks^{-\coalition} = [0,1]^{\edges^{-\coalition}}.$$
		For any $\networkbis^\coalition \in \networks^\coalition$, and any $\network^{-\coalition} \in \networks^{-\coalition}$, $(\networkbis^\coalition,\network^{-\coalition}) \in \networks$ is the network defined by $(\networkbis^\coalition,\network^{-\coalition})_{\coalition'} = (\network^{-\coalition})_{\coalition'}$, for any $\coalition' \in \edges^{-\coalition}$, and $(\networkbis^\coalition,\network^{-\coalition})_{\coalition'} = \networkbis^\coalition_{\coalition'}$, for any $\coalition' \in \edges^\coalition$.
		Last, for any network $\network \in \networks$, $\network^{-\coalition} \in \networks^{-\coalition}$ corresponds to the tuple of weights in $\network$ where the weights associated to members of $\coalition$ are missing.
	\end{itemize}
	
	We now introduce the notion of full stability within the framework of weighted network formation, incorporating, as in the previous concept, capacity constraints and budget constraints.
	
	\begin{definition}[Fully stable network] \label{def:full_stab}
		Let $\game = (\agents,\edges,\payoff, \capacity, \budget)$ be a network game.
		A network $\network \in \networks$ is \emph{fully stable with respect to $\game$} if the following three conditions hold:
		\begin{enumerate}
			\item $\network \in \networksab$.
			\item For every agent $i \in \agents$, and every $\networkbis^i \in \networks^i$, if, for every $\coalition \in \edges^i$, $\networkbis^i_\coalition \leq \network_\coalition$, then $\payoff_i(\networkbis^i,\network^{-i}) \leq \payoff_i(\network)$.
			\item For every coalition $\coalition \in \edges$, and every $\networkbis^\coalition \in \networks^\coalition$, if:
			\begin{enumerate}
				\item $\network_\coalition \leq \networkbis^\coalition_\coalition$;
				\item for every $\coalition' \in \edges^\coalition \backslash \{\coalition\}$, $\networkbis^\coalition_{\coalition'} \leq \network_{\coalition'}$;
				\item for every $i \in \coalition$, $\payoff_i(\networkbis^\coalition,\network^{-\coalition}) > \payoff_i(\network)$,
			\end{enumerate}
			then there exists $j \in \coalition$ such that
			$$\sum_{\coalition' \in \edges^j} \networkbis^\coalition_{\coalition'} > \budget_j.$$
		\end{enumerate}
	\end{definition}

The last (new) condition means that if a coalition can increase the payoffs of all its members by increasing their common weight -- with members possibly reducing the weights of some of their other incident edges simultaneously -- then at least one agent in the coalition must be violating its budget constraint.

 Finally, let us introduce some assumptions on payoff functions that will be needed for some existence and non-existence results.
    
    \begin{assumption}{S}{asspt:sep}
        A network game is called \emph{separable} if its payoff function $\payoff$ is separable, i.e., it can be written $\payoff = \sum_{\coalition \in \edges} \payoff_\coalition$, with $\payoff_\coalition: [0, \capacity_\coalition] \to \mathbb{R}^\agents$ for all $\coalition\in\edges$.
    \end{assumption}

    \begin{assumption}{C}{asspt:conc}
        A network game is called \emph{concave} if each agent $i$'s payoff function $\payoff_i$ is component-wise concave\footnote{This assumption is weaker that global concavity. Notice also that \ref{asspt:conc} implies \ref{asspt:quasiconc}.}.
    \end{assumption}

    \begin{assumption}{L}{asspt:lin}
        A network game is called \emph{linear} if each agent $i$'s payoff function $\payoff_i$ is linear\footnote{Notice that assumption \ref{asspt:lin} implies \ref{asspt:cont}, \ref{asspt:sep} and \ref{asspt:conc}.}.
    \end{assumption}

     \begin{assumption}{NE}{asspt:no_ext}
        A network game is said to have \emph{no externalities} if each agent $i$'s payoff function $\payoff_i$ verifies: $\forall \network \in \networksab$, $\forall \networkbis^{-i}\in \edges^{-i}$, $u_i(\network) = u_i(\network^i, \networkbis^{-i})$.
    \end{assumption}

    \begin{assumption}{PE}{asspt:pos_ext}
        A network game is said to have \emph{positive externalities} if each agent $i$'s payoff function $\payoff_i$ verifies: $\forall \network \in \networksab$, $\forall \networkbis^{-i}\in \edges^{-i}$ such that $\network^{-i} \succeq \networkbis^{-i}$ (i.e., $\network^{-i}$ is larger on all components), $u_i(\network) \geq u_i(\network^i, \networkbis^{-i})$.
    \end{assumption}

    \begin{assumption}{B}{asspt:bip}
        A network game is called \emph{bipartite} if $(\agents, \edges)$ forms a bipartite graph, i.e., if $\exists U, V \subset \agents$, $U\cap V = \emptyset$, such that $\agents = U \sqcup V$ and $\forall C\in\edges$, $C = \{i, j\}$ with $i\in U$ and $j\in V$.
    \end{assumption}
	
	\subsection{Link with Matching Theory}
	
	In this subsection, we express in network terms the notions of stability defined in matching theory, and the associated existence results. The definitions and results presented here are not used in the rest of the paper but they establish the link between our framework and fractional matching theory. While most of matching literature has been developed on graphs, it naturally extends to hypergraphs, and the definitions given below do not need $(\agents, \edges)$ to be a graph. 

    \begin{definition}[Matching problem]
        A \emph{fractional matching problem} is a network game that is linear \ref{asspt:lin} and has no externalities \ref{asspt:no_ext}.
    \end{definition}

    Preferences in matching problems are ordinal, i.e., each agent ranks their potential partners from best to worst. In our network formulation, it is equivalent to say that their payoff functions are linear (the ranking being induced by the coefficient of each variable), and without externalities since in matching only an agent's own connections matter to them.

    Many versions of matching problems have been given specific names. When the underlying graph is bipartite and all budgets are equal to 1 it is called the marriage problem, or one-to-one matching, and when one side has budgets of 1 and the other side's budgets are integers greater or equal to one it is called the college admission problem, or many-to-one matching. When budgets and edge capacities can be completely arbitrary, it is called ordinal transport \citep{Baiou_02}, and when the graph is not bipartite it is called the roommate problem. We define stability for those problems, by extending to our framework the definition proposed by Aharoni and Fleiner \cite{Aharoni_2003}.

    \begin{definition}[Stable matching]
    Let $\game = (\agents,\edges,\payoff, \capacity, \budget)$ be a fractional matching problem with $\payoff(\network) = (\sum_{\coalition \in \edges^i}A^i_\coalition \network_\coalition)_{i \in \agents}$. A network $\network \in \networksab$ is a \emph{stable matching} if for any edge $\coalition \in\edges$, 
    \begin{enumerate}
        \item if $\exists i \in \coalition$ such that $A^i_\coalition < 0$, $\network_\coalition = 0$, and
        \item if $\forall i \in \coalition$, $A^i_\coalition > 0$, $\exists i \in \coalition$ such that $\sum\limits_{\coalition' \ni i, A^i_{\coalition'} \geq A^i_\coalition} \network_{\coalition'} = \budget_i$.
    \end{enumerate}
    \end{definition}

    When $(\agents, \edges)$ is a hypergraph, this definition is slightly stronger than full stability, but both notions coincide on graphs.

    \begin{proposition} \label{lemma:stable_matching}
        Let $\game = (\agents,\edges,\payoff, \capacity, \budget)$ be a fractional matching problem. If $\network\in \networksab$ is a stable matching, then it is a fully stable network. If $\forall \coalition \in \edges$, $\vert \coalition \vert = 2$, the converse is also true.
    \end{proposition}

    The proof is provided in Appendix \ref{app.proof:stable_matching}.

    Many existence results have been proposed for different versions on the stable fractional matching problem, we recall here the two most general ones.

    \begin{theorem*}
        Let $\game = (\agents,\edges,\payoff, \capacity, \budget)$ be a fractional matching problem, i.e., a linear \ref{asspt:lin} network game with no externalities \ref{asspt:no_ext}.
        \begin{enumerate}
            \item (Baïou--Balinski \cite{Baiou_02}). If $(\agents, \edges)$ is a bipartite graph, then $\game$ admits a stable matching (and thus a fully stable network).
            \item (Aharoni--Fleiner \cite{Aharoni_2003}). If $\budget = (1, \dots , 1)$, then $\game$ admits a stable matching (and thus a fully stable network).
      \end{enumerate}
    \end{theorem*}

    Those two results are complementary: Baïou--Balinski allows any constraint (integral or not) but is restricted to bipartite graphs, while Aharoni--Fleiner is true for any hypergraph but only when all agents have a budget of 1.

	\subsection{Two Counter-examples to the Existence of a Fully Stable Network}

    \begin{theorem} \label{thm:counter}
        There exists a linear \ref{asspt:lin} and bipartite \ref{asspt:bip} unconstrained network game $\game = (\agents,\edges,\payoff)$ that admits no fully stable network.
    \end{theorem}

    \begin{proof}
	   Consider the following network game.
	   The set of agents is $\agents = \{1,2,3\}$, and the set of edges is $\edges = \{\{1,2\},\{1,3\}\}$.
	   Payoffs are given by
       \begin{align*}
           \payoff_1(\network) &= -\network_{\{1,2\}}-2\network_{\{2,3\}}, \\
           \payoff_2(\network) &= 2\network_{\{1,2\}}+\network_{\{2,3\}}, \\
           \payoff_3(\network) &= \network_{\{2,3\}}.
       \end{align*}
	   There are neither capacity constraints nor budget constraints in this example. Note also that the underlying graph $(\agents,\edges)$ is bipartite.
	
	   We claim that the game $(\agents,\edges,\payoff)$ admits no fully stable network.
	   To see this, observe first that any candidate for full stability must coincide with the network $(\network_{\{1,2\}},\network_{\{2,3\}}) = (0,1)$.
	   Indeed, agent 1 strictly prefers the weight on edge $\{1,2\}$ to be equal to 0, while agents 2 and 3 have a common interest in setting the weight on edge $\{2,3\}$ equal to 1.
	   This network is therefore the only pairwise stable network in the sense of Bich and Morhaim, and the only candidate for full stability.
	
	   However, full stability fails.
	   Starting from the network $(0,1)$, the payoff vector is $(-2,1,1)$.
	   Consider now the network $(1,0)$.
	   At first glance, this deviation runs counter to the intuition of pairwise stability: agents 1 and 2 increase the weight on edge $\{1,2\}$, even though agent 1 would not benefit from such an increase in isolation, while agent 2 simultaneously decreases the weight on edge $\{2,3\}$, despite normally benefiting from it.
	   Nevertheless, this coordinated reallocation yields the payoff vector $(-1,2,0)$, so that both agents 1 and 2 strictly improve relative to $(0,1)$.    
    \end{proof}
	
	The economic intuition is that agents 1 and 2 reach a mutually beneficial agreement.
	Although agent 1 has no intrinsic incentive to increase the weight on $\{1,2\}$, they accept this increase because it comes together with a decrease in the weight on $\{2,3\}$, to which they are much more sensitive.
	Symmetrically, agent 2 is willing to weaken their relationship with agent 3 in exchange for a stronger connection with agent 1.
	This deviation is therefore profitable for both agents 1 and 2, even though neither unilateral move would be individually attractive.
	
	The instability persists.
	Starting from the network $(1,0)$, agents 2 and 3 have a common interest in increasing the weight on edge $\{2,3\}$, which leads to the network $(1,1)$.
	At this point, the initial agreement between agents 1 and 2 is no longer respected.
	Agent 1 then strictly prefers to sever the connection with agent 2, returning to the network $(0,1)$.
	This generates a cycle of profitable joint deviations, showing that no network can be fully stable in this game:
	$$(0,1) \longrightarrow_{\{1,2\}} (1,0) \longrightarrow_{\{2,3\}} (1,1) \longrightarrow_1 (0,1) \longrightarrow_{\{1,2\}} \dots$$
	
	This counter-example highlights the role played by negative externalities in the non-existence of a fully stable network.
	In particular, a cycle such as the one described above would not be possible without the negative externality faced by agent 1 with respect to the connection between agents 2 and 3. However, negative externalities are not the only obstacle to full stability.

    \begin{theorem} \label{thm:counter_2}
        There exists a continuous \ref{asspt:cont}, quasiconcave \ref{asspt:quasiconc}, separable \ref{asspt:sep} network game $\game = (\agents,\edges,\payoff, \capacity, \budget)$ with no externalities \ref{asspt:no_ext} that admits no fully stable network.
    \end{theorem}

    \begin{proof}
        Consider the following network game.
        The set of agents is $\agents = \{1,2,3\}$, the set of edges is $\edges = \{\{1,2\},\{1,3\}, \{2,3\}\}$, and the budget vector is $\budget = (1, 1, 1)$.
        Payoffs are given by
        \begin{minipage}{0.34\textwidth}
            \begin{align*}
            \payoff_1(\network) &= f(\network_{\{1,2\}}) + \network_{\{1,3\}}, \\
            \payoff_2(\network) &= f(\network_{\{2,3\}}) + \network_{\{1,2\}}, \\
            \payoff_3(\network) &= f(\network_{\{1,3\}}) + \network_{\{2,3\}},
        \end{align*}
        \end{minipage} 
        \begin{minipage}{0.65\textwidth}
\begin{figure}[H]
    \includegraphics[width = 0.7\linewidth]{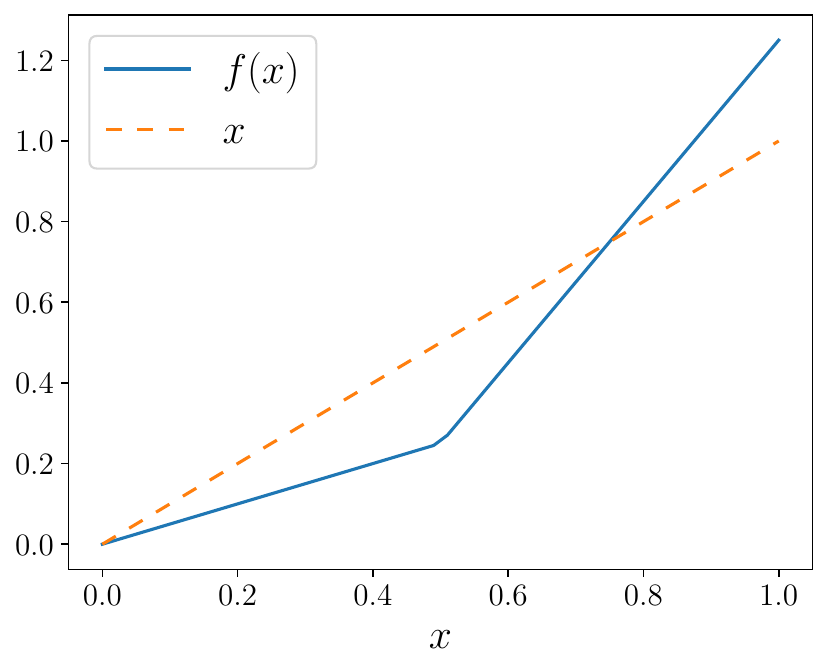}\vspace{-0.4cm}
    \caption{Graphs of the two functions that compose the payoff of agents in the counter-example used in the proof of \Cref{thm:counter_2}.}
    \label{fig:counter_2}
\end{figure} 
\end{minipage} 

        where $f(x):[0, 1] \to \mathbb{R}$ is defined by $f(x) = (x/2)\mathbbm{1}_{x \leq 1/2} + (2x-3/4)\mathbbm{1}_{x\geq 1/2}$ and is illustrated in \Cref{fig:counter_2}.

        We claim that this game admits no fully stable network. Let $\network = (\network_{\{1, 2\}}, \network_{\{1, 3\}}, \network_{\{2, 3\}}) \in \networksab$. We reason on the number of non-zero components of $\network$. If $\network = 0_{\mathbb{R}^3}$, since all payoffs are increasing, any weight increase is beneficial to all agents involved, so $0_{\mathbb{R}^3}$ is not fully stable. If only one weight is non-zero, e.g., $\network = (\network_{\{1, 2\}}, 0, 0)$, then agents $2$ and $3$ are both better off deviating to $(0, 0, 1)$, so no network with only one non-zero weight can be fully stable. If two weights at least are non-zero, e.g., $\network_{\{1, 2\}}$ and $\network_{\{1, 3\}}$, by the budget constraints, one of them has to be not greater than $1/2$. If $\network_{\{1, 2\}} \leq 1/2$, $\payoff_1(\network)<1$. Moreover, $\network_{\{1, 3\}} \leq 1 - \network_{\{1, 2\}}<1$, so $\payoff_3(\network)<5/4$. Therefore, $\network' =(0, 1, 0)$ gives $\payoff_1(\network') = 1 > \payoff_1(\network)$ and $\payoff_3(\network') = 5/4 > \payoff_3(\network)$, so $\network$ is not fully stable. By a similar argument, if $\network_{\{1, 3\}} \leq 1/2$, both agents $2$ and $3$ are better off with the vector $\network' =(0, 0, 1)$. In conclusion, this game admits no fully stable network.
    \end{proof}

Observe that each player payoff function is, in addition of being without externalities, quasi-concave and separable, it is in addition convex and increasing. Thus, here, the reason for the non-existence of fully stable networks is not the externalities (which are absent), but the non-concavity of some payoff functions. Agents have weights that are very beneficial to them if they are high, but uninteresting if they are low, what economists call \emph{increasing returns}. For this reason, interior networks (i.e., networks with $\network_{\{i, j\}} \in (0, 1)$) are inefficient and blocked by networks on the boundary. If, in addition, boundary networks form a profitable deviation cycle as in the the proof:
    $$(1, 0, 0) \longrightarrow_{\{2,3\}} (0, 0, 1) \longrightarrow_{\{1,3\}} (0, 1, 0) \longrightarrow_{\{1,2\}} (1, 0, 0) \longrightarrow_{\{2,3\}} \dots,$$
    implying that no fully stable network can exist in this example.

    \section{Full Stability: Existence Results} \label{sec:full_exist}

    In this section, we present two existence results that avoid the issues of the two counter-examples presented above by making additional assumptions on the payoff functions.
    
	\subsection{The Hypergraphic Case without Constraints and with Positive Externalities}

    From \Cref{thm:counter} we know that the presence of negative externalities can prevent the existence of fully stable networks, even without budget constraints and very strong assumptions on the shape of the graph or the payoff functions (e.g., linearity). However, we show now that when externalities are assumed to be positive we can recover existence under separability.
    
    \begin{theorem} \label{thm:hypergraph_full_wo_cnst}
        Let $\game = (\agents,\edges,\payoff)$ be a continuous \ref{asspt:cont} and quasiconcave \ref{asspt:quasiconc} unconstrained network game.
        Suppose that the following condition is fulfilled: for every stable network $\network \in \networks$, if there exists $\networkbis^\coalition \in \networks^\coalition$ such that
			\begin{enumerate}
				\item $\network_\coalition \leq \networkbis^\coalition_\coalition$;
				\item for every $\coalition' \in \edges^\coalition \backslash \{\coalition\}$, $\networkbis^\coalition_{\coalition'} \leq \network_{\coalition'}$;
				\item for every $i \in \coalition$, $\payoff_i(\networkbis^\coalition,\network^{-\coalition}) > \payoff_i(\network)$,
			\end{enumerate}
        then there exists $z_\coalition \in \networks_\coalition$ such that $\network_\coalition \leq z_\coalition$ and, for every $i \in \coalition$, $\payoff_i(z_\coalition,\network_{-\coalition}) > u_i(\network)$.
        Then the set of stable networks and the set of fully stable networks of $~\game$ coincide.
        In particular, $\game$ admits a fully stable network.
    \end{theorem}

    \begin{proof}
        First of all, it is immediate to see that any fully stable network is also stable.
        
        On the other hand, suppose that $\network$ is stable, but not fully stable.
        By definition, this directly implies that there exists $\networkbis^\coalition \in \networks^\coalition$ such that \begin{enumerate}
				\item $\network_\coalition \leq \networkbis^\coalition_\coalition$;
				\item for every $\coalition' \in \edges^\coalition \backslash \{\coalition\}$, $\networkbis^\coalition_{\coalition'} \leq \network_{\coalition'}$;
				\item for every $i \in \coalition$, $\payoff_i(\networkbis^\coalition,\network^{-\coalition}) > \payoff_i(\network)$.
			\end{enumerate}
        Hence by assumption, there exists $z_\coalition \in \networks_\coalition$ such that $\network_\coalition \leq z_\coalition$ and, for every $i \in \coalition$, $\payoff_i(z_\coalition,\network_{-\coalition}) > \payoff_i(\network)$.
        This means precisely that $\network$ is not stable, a contradiction.
        
        Finally, note that by \Cref{thm:hypergraph_stable} (considering neither capacity constraints nor budget constraints), the set of stable networks of $\game$ is nonempty.
    \end{proof}

    The condition in the previous theorem states that if the members of a coalition have an incentive to increase the weight of their common link, possibly while simultaneously decreasing other links, then they must in fact have an incentive to increase their common weight even by a small amount, without any additional decreases.
    
    In particular, any network game that is separable \ref{asspt:sep} and has positive externalities \ref{asspt:pos_ext} verifies the assumption of \Cref{thm:hypergraph_full_wo_cnst}.
    However, in that case, we can go further and compute explicitly all fully stable matchings.
    
	\begin{proposition} \label{prop:min_argmax}
        Let $\game = (\agents,\edges,\payoff)$ be a continuous \ref{asspt:cont}, separable \ref{asspt:sep}, and quasiconcave \ref{asspt:quasiconc} unconstrained network game with positive externalities \ref{asspt:pos_ext}. Then the fully stable networks of $\game$ are exactly the networks that satisfy
        \begin{equation} \label{eq:min_argmax}
            \forall \coalition \in \edges, \network_\coalition \in  \left[\min\limits_{i \in \coalition}\left(\min (\argmax \payoff_i(\cdot, x_{-\coalition})) \right), \min\limits_{i \in \coalition}\left(\max( \argmax \payoff_i(\cdot, x_{-\coalition}) ) \right)   \right] .
        \end{equation}
        In particular, if each agent's payoff $\payoff_i$ is component-wise \emph{strictly} quasiconcave, then equation \eqref{eq:min_argmax} gives the unique fully stable network of $\game$.
    \end{proposition}

    \begin{proof}[Sketch of proof.]
        In the absence of budget constraints and under the assumptions of the proposition, each weight $\network_\coalition$ can be increased until it reaches the argmax of the partial payoff function of one of $\coalition$'s members. If $\network_\coalition$ is smaller than this, all members of $\coalition$ will agree to increase it so $\network$ would be unstable. On the other hand, if $\network_\coalition$ goes above the argmax of the partial payoff function of some member of $\coalition$, this agent will have an incentive to decrease it. The full proof is provided in \Cref{app.proof:min_argmax}.
    \end{proof}

    \begin{remark}
        The inclusion of the set described by equation \eqref{eq:min_argmax} inside the set of pairwise stable networks was already proved in a lemma in the proof of the main theorem of Bich and Morhaim  (see Lemma A.2 in \cite{bich_morhaim_2020}), but only for graphs (as they did not consider hypergraphs).
    \end{remark}

    A natural question is now whether, assuming positive externalities \ref{asspt:pos_ext}, the separability \ref{asspt:sep} and unconstrained-ness assumptions can be lifted. We have a partial answer: from \Cref{thm:counter_2}, we know that adding constraints in this setting can prevent existence. However, the question of the existence of a fully stable network in continuous \ref{asspt:cont}, quasiconcave \ref{asspt:quasiconc}, \emph{non-separable}, unconstrained network games with positive externalities \ref{asspt:pos_ext} is still open.
    
	\subsection{The Bipartite Case with Constraints and with Positive Externalities}

    From \Cref{thm:counter_2}, we know that payoffs that are not concave can prevent the existence of a fully stable network when there are budget constraints. To bypass this issue, we propose an algorithm for bipartite graphs called \emph{\alg} (\algshort), that generalizes the \emph{Row-Greedy Algorithm} of Baïou and Balinski \cite{Baiou_02}, itself generalizing the \emph{Deferred Acceptance Algorithm} of Gale and Shapley \cite{gale_college_1962}. For any continuous \ref{asspt:cont}, separable \ref{asspt:sep} and concave \ref{asspt:conc} network game without externalities \ref{asspt:no_ext} on a bipartite graph, \algshort~  converges to a fully stable network, proving that those assumptions are sufficient to recover existence. We further prove that this algorithm can be extended to payoffs with positive externalities \ref{asspt:pos_ext}.

    Algorithm \ref{alg:GDA} describes the full algorithm; let us informally explain it. We assume that $(\agents, \edges)$ is a bipartite graph with $\agents = U \sqcup V$. \algshort~ tries to give each $U$ agent the weights that maximize their payoff (this can be done separately for each $U$ agent since the graph is bipartite and there are no externalities). It then submits the obtained network to $V$ agents, who can decrease the weights of links they are a part of in the way that maximizes their payoff. If an agent $j \in V$ decreases their connection with an agent $i \in U$, $i$ is forbidden of ever increasing it again. This describes one full step of the algorithm. We then start again from the obtained network, $U$ agents increase the weights of links they are a part of in a way that maximizes their payoff, and then $V$ agents decrease them, and so on. The algorithm stops when $U$ agents cannot increase their payoff. We can prove that this algorithm converges to a fully stable network, which is sufficient to prove its existence. However, it might require an infinite number of steps to converge. When applied to a fractional matching problem, the algorithm reduces to the Row-Greedy algorithm of Baïou and Balinski \cite{Baiou_02}.

    \begin{algorithm}[ht]
        \SetAlgoLined
        \KwData{Network game $\game = (\agents,\edges,\payoff, \capacity, \budget)$ with $(\agents = U \sqcup V, \edges)$ a bipartite graph}
        \KwResult{Fully stable network $\network$}
        $\network \leftarrow 0_{\mathbb{R}^{\edges}}$\;
        \While{True}{
        $\network' \leftarrow 0_{\mathbb{R}^{\edges}}$\;
        \For{$i \in U$}{
        $\network'^{i} \leftarrow \argmax\limits_{\substack{\networkbis^i \in \prod_{\coalition \ni i}[\network_\coalition, \capacity_\coalition] \\ \sum_{\coalition \ni i} \networkbis^i_\coalition \leq \budget_i}} \payoff_i(\networkbis^i)$\;
        If the argmax contains several elements, if $0_{\mathbb{R}^{\edges^j}}$ is one of them choose it, otherwise choose one arbitrarily\;}
        \If{$\network' = 0_{\mathbb{R}^{\edges}}$}{
        \textbf{Terminate and return $\network$.}
        }
        \For{$j \in V$}{
        $\network^{j} \leftarrow \argmax\limits_{\substack{\networkbis^j \in \prod_{\coalition \ni j} [0, \network'_\coalition]\\ \sum_{\coalition \ni j} \networkbis^j_\coalition \leq \budget_j}} \payoff_j(\networkbis^j)$ \;
        If the argmax contains several elements, if $0_{\mathbb{R}^{\edges^j}}$ is one of them choose it, otherwise choose one arbitrarily\;}
        \For{$\coalition \in \edges$}{
        \If{$\network_\coalition < \network'_\coalition$}{
        $\capacity_\coalition \leftarrow \network_\coalition$\;
        }
        }
        \If{Some $\capacity_C$ has been modified}{Continue (go back to \textbf{repeat} without checking end condition)\;}
        }
        \caption{\alg}
        \label{alg:GDA}
    \end{algorithm}

    \begin{theorem} \label{thm:bipartite_full}
    Let $\game = (\agents,\edges,\payoff, \capacity, \budget)$ be a continuous \ref{asspt:cont} and bipartite \ref{asspt:bip} network game.
        \begin{enumerate}
            \item If $\game$ has no externalities \ref{asspt:no_ext}, then \algshort~ converges to a network $\network \in \networksab$.
            \item If, in addition, $\game$ is separable \ref{asspt:sep} and concave \ref{asspt:conc}, then $\network$ is fully stable.
        \end{enumerate}
    \end{theorem}

    \begin{proof}[Sketch of proof.]
        (1) After each full step, either some agent on the disposing side ($V$ in the algorithm) has strictly increased their payoff, or some capacities have been reduced. Then, the sum of the payoffs of agents on the disposing side minus the sum of the capacities is a strictly increasing and bounded quantity, so the algorithm converges.

        (2) Agents on the proposing side ($U$ in the algorithm) maximize their payoff at each step. Separability and concavity ensure that they do not regret having allocated weight to a connection at a prior step even if it has been decreased by its other member. Then, for any profitable deviation they have consisting in increasing some weight and (maybe) decreasing others, the weight increase has already been proposed and refused, and by concavity this implies that the deviation is not profitable to the other agent even if the rejection happened at an earlier step.

        The full proof is provided in \Cref{app.proof:bipartite_full}.
    \end{proof}

    Baïou and Balinski \cite{Baiou_02} \ left as an open question whether the Row-Greedy algorithm would converge to a stable matching (if it terminated) when the constraints were not integral; \Cref{thm:bipartite_full} confirms that is does.

    The algorithm can in fact be used to compute fully stable networks even in the presence of positive externalities, as show below.

    \begin{corollary} \label{cor:bipartite_pos_ext}
        Let $\game = (\agents,\edges,\payoff, \capacity, \budget)$ be a continuous \ref{asspt:cont}, separable \ref{asspt:sep} and concave \ref{asspt:conc} network game with positive externalities \ref{asspt:pos_ext}. Then $\game$ admits a fully stable network, and it can be computed by running \algshort~ on the game $\game' = (\agents,\edges,(\payoff(\cdot, \network^{-i}))_{i \in \agents}, \capacity, \budget)$.
    \end{corollary}

    To prove it, we need the following result.

    \begin{lemma}\label{prop:pos_ext}
        Let $\game = (\agents,\edges,\payoff, \capacity, \budget)$ be a network game  with positive externalities \ref{asspt:pos_ext} such that for all $i\in\agents$, $\payoff_i(\network) = \payoff^{int}_i(\network^i) + \payoff^{ext}_i(\network^{-i})$\footnote{This assumption is much weaker than separability, which applies to all variables.}. Then any fully stable network of the game $\game' = (\agents,\edges,\payoff^{int}, \capacity, \budget)$ is also a fully stable network of \  $\game$.
    \end{lemma}
    
    \begin{proof}
        Let $\network \in \networksab$ be a fully stable network of $\game'$, and suppose it is not stable for $\game$. Then, there exists $\coalition \in \edges$ and $\networkbis^\coalition \in \networks^\coalition$ that constitutes an admissible deviation according to \Cref{def:full_stab} and such that $\forall i \in \coalition, \payoff_i(\networkbis^\coalition, \network^{-\coalition}) > \payoff_i(\network)$. Let us write $z = (\networkbis^\coalition, \network^{-\coalition})$. Then, for any $i \in \coalition$,
        \begin{align}\label{eq:proof_pos_ext}
            \payoff_i(z) > \payoff_i(\network) &\Leftrightarrow \payoff_i(z^i, z^{-i}) > \payoff_i(\network^i, \network^{-i}) \nonumber \\ 
             & \Rightarrow \payoff_i(z^i, \network^{-i}) > \payoff_i(\network^i, \network^{-i}) \\
             & \Rightarrow \payoff^{int}_i(z^i) > \payoff^{int}_i(\network^i), \nonumber
        \end{align}
        which contradicts the full stability of $\network$ for $\game'$. The second line of equation \eqref{eq:proof_pos_ext} comes from the positive externalities assumption and that $z^{-i} \preceq \network^{-i}$ since $\networkbis^\coalition$ is an admissible deviation for $\coalition$.
    \end{proof}
    
    \begin{proof}[Proof of \Cref{cor:bipartite_pos_ext}]
        By \Cref{thm:bipartite_full}, \algshort~ produces a fully stable network of $\game'$. By \Cref{prop:pos_ext}, this network is also fully stable for $\game$.
    \end{proof}

    Once again, since we made the assumptions that the game is bipartite \ref{asspt:bip}, separable \ref{asspt:sep} and concave \ref{asspt:conc} (rather than quasiconcave), a natural question is whether those assumptions are simultaneously necessary. This is still an open question, however, we can prove that \algshort~ fails without any of those assumptions. First, if $(\agents, \edges)$ is not a bipartite graph, the algorithm is not well defined. Secondly, the two examples below show that the algorithm can converge to an unstable network if the game is not simultaneously separable and concave.

    \begin{example}
        Consider the problem $\agents = \{1, 2, 3\}, \edges = \{\{1, 2\}, \{1, 3\}\}$, with payoffs defined as
        $$\begin{aligned}
            \payoff_1(\network) &= \network_{\{1,2\}} +  \network_{\{1,3\}} - 3(\network_{\{1,2\}} - \network_{\{1,3\}})^2, \\
            \payoff_2(\network) &= \network_{\{1,2\}} (1 - \network_{\{1,2\}}), \\
            \payoff_3(\network) &=  \network_{\{1,3\}},
        \end{aligned}$$
        capacities all equal to 1 and budgets all equal to 2. Notice that $\payoff$ is continuous, concave, has no externalities, but is not separable. Let $U = \{1\}$ and $V = \{2, 3\}$. At the first step of the algorithm, agent $1$ computes their argmax, which is $\network_{\{1,2\}} = \network_{\{1,3\}} = 1$. Then, agent $2$ decreases $\network_{\{1,2\}}$ to $1/2$, and agent $3$ keeps $\network_{\{1,3\}}$ at 1. At the next step, the algorithm terminates. However, $\network = (1/2, 1)$ is not fully stable: $\payoff_1(\network) = 3/4$, and agent $1$ could increase their payoff to 1 by decreasing $\network_{\{1,3\}}$ to $1/2$.
    \end{example}

    \begin{example}
        Consider the problem $\agents = U \sqcup V, U = \{1, 2\}, V = \{3, 4, 5\}, \edges = \{\{i, j\}:i \in U, j \in V\}$, with payoffs defined as
        $$\begin{aligned}
            \payoff_1(\network) &= f_{2/3}(\network_{\{1,3\}}) + \network_{\{1,4\}} + (3/4) \network_{\{1,5\}}, \\
            \payoff_2(\network) &= f_{2/3}(\network_{\{2,4\}}) + \network_{\{2,3\}} + (3/4)  \network_{\{2,5\}}, \\
            \payoff_3(\network) &=  f_{1/3}(\network_{\{2,3\}}) + \network_{\{1,3\}}, \\
            \payoff_4(\network) &=  f_{1/3}(\network_{\{1,4\}}) + \network_{\{2,4\}},\\
            \payoff_5(\network) &=  2\network_{\{1,5\}} + \network_{\{2,5\}},
        \end{aligned}$$
        where $f_{y}(x):[0, 1] \to \mathbb{R}$ is defined by $f(x) = (x/2)\mathbbm{1}_{x \leq y} + (2x-3y/2)\mathbbm{1}_{x\geq y}$ (as an illustration, the function $f$ diplayed in \Cref{fig:counter_2} corresponds to $f_{1/2}$). All capacities are 1, $\budget = (3/2, 3/2, 1, 1, 1)$. Notice that $\payoff$ is continuous, quasiconcave (even convex and increasing), separable and has no externalities but is not concave. At the first step of the algorithm, agents $1$ and $2$ propose the weights $(\network_{\{1, 3\}}, \network_{\{1, 4\}}, \network_{\{1, 5\}}, \network_{\{2, 3\}}, \network_{\{2, 4\}}, \network_{\{2, 5\}}) = (1, 1/2, 0, 1, 1/2, 0)$, which are then reduced by agents $4$ and $5$ to $(1/2, 1/2, 0, 1/2, 1/2, 0)$. At the second step, $1$ and $2$ propose to increase to $(1/2, 1, 0, 1/2, 1, 0)$, which is once again reduced to $(1/2, 1/2, 0, 1/2, 1/2, 0)$. At the third step, $1$ and $2$ finally propose to increase to $\network = (1/2, 1/2, 1/2, 1/2, 1/2, 1/2)$, which is accepted as is, and the algorithm terminates. However, $\network$ is not fully stable: agents $1$ and $5$ get respective payoffs $\payoff_1(\network) = 9/8$ and $\payoff_5(\network) = 3/2$, while they could deviate to $\network' = (0, 1/2, 1, 1/2, 1/2, 0)$ and get respectively $\payoff_1(\network) = 5/4$ and $\payoff_5(\network) = 2$.
    \end{example}

    Finally, this algorithm can also be used on discrete problems, for instance many-to-many matching problems, but also more general problems where the marginal payoff of increasing the weight of a link by one unit is non-increasing rather than constant.

    \begin{proposition} \label{prop:discrete_matching}
        Let $\game = (\agents,\edges,\payoff, \capacity, \budget)$ be a continuous \ref{asspt:cont}, separable \ref{asspt:sep}, concave \ref{asspt:conc} and bipartite \ref{asspt:bip} network game with no externalities \ref{asspt:no_ext} and integral capacities and constraints\footnote{We assumed so far that the capacities $\capacity$ were less than 1 without loss of generality and for simplicity, however here we allow them to be larger than 1.}. Assume that payoffs are piecewise affine, such that slope changes only happen at integral points. If \algshort~ is run with the added rule that each time an element is chosen from an $\argmax$, it has to be integral, then it produces a fully stable integral network in at most $\sum_{\coalition \in \edges} \capacity_\coalition$ steps.
    \end{proposition}

    \begin{proof}
        Under the assumptions made in the statement of the proposition, at each step, all capacity constraints are integral (integer value), as well as the remaining budget of each agent. Since payoffs are piecewise affine, all the $\argmax$ that are computed must therefore contain integral vectors. Then, throughout the algorithm, all weights remain integral, and thus the output is also integral. Moreover, the sum of all capacities decreases by at least 1 at each step, so the algorithm terminates in at most $\sum_{\coalition \in \edges} \capacity_\coalition$ steps.
    \end{proof}

    \section{Discussion} \label{sec:disc}

Our framework encompasses many existing models of strategic network formation as well as stable matching problems. We establish the existence of a stable weighted hypergraphic network under budget constraints and mild assumptions on payoff functions. Regarding full stability, we either provide existence results for fully stable networks or counter-examples showing that their existence is not guaranteed. Nevertheless, the theory is not yet complete, and several open problems remain.

  Without budget constraints, negative externalities preclude existence even for linear payoff functions and bipartite graphs. The existence of a fully stable network is recovered when payoff functions are separable, quasiconcave and continuous and exhibit only positive externalities. However, it remains unclear whether this condition is necessary; for instance, does existence still hold for non-separable quasiconcave payoffs with positive (or without) externalities?.   
  
With budget constraints, separability and positive externalities do not suffice to guarantee the existence of a fully stable network, even when payoffs are monotonic and without externalities. However, we are able to prove existence under separability, concavity, and positive externalities, but only for bipartite graphs. Does this result extend to general graphs or hypergraph network games?

 Also, note that the \alg \ algorithm is sufficient to establish the existence of a fully stable network under budget constraints, separability, concavity, and positive externalities of the payoff functions. However, it may require an infinite number of steps to converge. The existence of an algorithm that computes a fully stable network in finite (or polynomial) time remains an open question. There are good reasons to believe that a more efficient algorithm than ours exists. Indeed, our algorithm is a generalization of Baïou and Balinski’s Row-Greedy algorithm; however, in the same article they propose a more efficient (i.e., polynomial-time) algorithm—albeit more complex—called the Inductive Algorithm. It is unclear how this algorithm could be extended to the non-linear setting.

 Finally, the question of designing algorithms for the non-bipartite case—namely, general graphs or hypergraphs—remains open. With budget constraints, linear payoffs, and no externalities, Abeledo and Rothblum \cite{abeledo_94} propose a reduction from non-bipartite graphs to the bipartite case that allows one to compute a fully stable network in polynomial time; for hypergraphs, the problem is known to be PPAD-complete \cite{kintali_ppad_13}. In the presence of non-linear payoffs with or without positive externalities, the existence of fully stable network and of an (efficient) algorithm to compute one of them remains an open problem.

    \bibliography{biblio}
    
    \appendix

	\section{Omitted Proofs}
    
	\subsection{Proof of \texorpdfstring{\Cref{thm:hypergraph_stable}}{}} \label{app.proof:hypergraph_stable}
	
	Consider the correspondence $\Psi: \networksab \to 2^{\networksab}$ defined as follows.
	For any $\network,\network' \in \networksab$, we have $\network' \in \Psi(\network)$ if and only if there exists a coalition $\coalition \in \edges$ such that one of the following conditions holds:
	\begin{enumerate}
		\item There exists a weight $\networkbis_\coalition \in [0,\network_\coalition[$ such that $\network' = (\networkbis_\coalition,\network_{-\coalition})$ and
			$$\exists i \in \coalition, \payoff_i(\network') > \payoff_i(\network).$$
		\item There exists a weight $\networkbis_\coalition \in ]\network_\coalition,\capacity_\coalition]$ such that $\network' = (\networkbis_\coalition,\network_{-\coalition})$ and
			$$[\forall i \in \coalition, \payoff_i(\network') > \payoff_i(\network)] \wedge [\forall i \in \coalition, \sum_{\substack{\coalition' \neq \coalition\\\coalition' \ni i}} \network_{\coalition'} + \networkbis_\coalition \leq \budget_i].$$
	\end{enumerate}
	We also recall the following result due to Yannelis and Prabhakar (see Theorem 5.2 in \cite{yannelis_prabhakar_1983}).
	
	\begin{theorem}[Yannelis--Prabhakar \cite{yannelis_prabhakar_1983}] \label{thm_yannelis_prabhakar}
		Let $K$ be a nonempty, compact and convex subset of $\mathbb{R}^n$.
		Moreover, let $F: K \to 2^K$ be a lower hemicontinuous correspondence such that, for every $\network \in K$, $\network \notin \mathrm{Conv}(F(\network))$.
		Then there exists $\network^\ast \in K$ such that $F(\network^\ast) = \emptyset$.
	\end{theorem}
		
	Since $\networksab$ is nonempty, compact and convex, it remains to show that the correspondence $\Psi$ satisfies the assumptions of Theorem \ref{thm_yannelis_prabhakar}.
		
	\paragraph{Step 1. $\Psi$ is lower hemicontinuous}
		
	Since $\networksab$ is a subset of a finite-dimensional Euclidean space, recall that the correspondence $\Psi$ is lower hemicontinuous if and only if for every $\network \in \networksab$, every sequence $(\network^n)_{n \in \mathbb{N}}$ converging to $\network$, and every $\network' \in \Psi(\network)$, there exists a subsequence $(\network^{n_m})_{m \in \mathbb{N}}$ and a sequence $((\network')^m)_{m \in \mathbb{N}}$ such that $(\network')^m \longrightarrow \network'$ and $(\network')^m \in \Psi(\network^{n_m})$ for all $m$.
		
	Fix $\network \in \networksab$, a sequence $(\network^n)_{n \in \mathbb{N}}$ converging to $\network$, and $\network' \in \Psi(\network)$.
	Moreover, by definition of $\Psi$, there exists a coalition $\coalition \in \edges$ such that either condition $(1)$ or condition $(2)$ holds.
		
	\paragraph{Case 1: unilateral decrease}
		
	Assume condition $(1)$ holds.
	Then there exists $\networkbis_\coalition \in [0,\network_\coalition[$ such that $\network' = (\networkbis_\coalition,\network_{-\coalition})$ and
		$$\exists i \in \coalition, \payoff_i(\network') > \payoff_i(\network).$$
		
	Since $\networkbis_\coalition < \network_\coalition$, there exists $\bar n_1$ such that $\networkbis_\coalition < \network^n_\coalition$ for all $n \geq \bar n_1$.
		
	By continuity of $\payoff_i$, there exists $\bar n_2$ such that $\payoff_i(\networkbis_\coalition,\network^\ell_{-\coalition}) > \payoff_i(\network^\ell)$ for all $n \geq \bar n_2$.
		
	Let $\bar n = \max \{\bar n_1,\bar n_2\}$, define $\network^{n_m} = \network^{\bar n + m}$, and set
	$$(\network')^m = (\networkbis_\coalition,\network^{n_m}_{-\coalition}).$$
		
	Since $\networkbis_\coalition < \network^{n_m}_\coalition$, feasability with respect to $\capacity$ and $\budget$ is preserved, so $(\network')^m \in \networksab$ for all $m$.
	Moreover, by construction, $(\network')^m$ belongs to $\Psi(\network^{n_m})$ for all $k$.
		
	\paragraph{Case 2: multilateral increase under constraints}
		
	Assume condition $(2)$ holds.
	Then there exists $\networkbis_\coalition \in ]\network_\coalition,\capacity_\coalition]$ such that $\network' = (\networkbis_\coalition,\network_{-\coalition})$ and
	$$[\forall i \in \coalition, \payoff_i(\network') > \payoff_i(\network)] \wedge [\forall i \in \coalition, \sum_{\substack{\coalition' \neq \coalition\\\coalition' \ni i}} \network_{\coalition'} + \networkbis_\coalition \leq \budget_i].$$

	Define
	$$\agents^<_\coalition = \{i \in \coalition : \sum_{\substack{\coalition' \neq \coalition\\\coalition' \ni i}} \network_{\coalition'} + \networkbis_\coalition < \budget_i\} \text{ and } \agents^=_\coalition = \{i \in \coalition : \sum_{\substack{\coalition' \neq \coalition\\\coalition' \ni i}} \network_{\coalition'} + \networkbis_\coalition = \budget_i\}.$$
	(If one of these sets is empty, the argument simplifies and we omit it.)
		
	Moreover, let
	$$\hat \networkbis^n_\coalition = \min \{\networkbis_\coalition, \min \{\budget_i - \sum_{\substack{\coalition' \neq \coalition\\\coalition' \ni i}} \network^n_{\coalition'} : i \in \agents^=_\coalition\}\}.$$
		
	Since $\network_\coalition < \networkbis_\coalition$, there exists $\bar n_1$ such that $\network^n_\coalition < \networkbis_\coalition$ for all $n \geq \bar n_1$.
		
	For every $i \in \agents$, by continuity of $\payoff_i$, there exists $\bar n_{2,i}$ such that $\payoff_i(\hat \networkbis^n_\coalition,\network^n_{-\coalition}) > \payoff_i(\network^n)$ for all $n \geq \bar n_{2,i}$.

	For every $i \in \agents^<_\coalition$, since $\sum_{\substack{\coalition' \neq \coalition\\\coalition' \ni i}} \network_{\coalition'} + \networkbis_\coalition < \budget_i$, there exists $\bar n_{3,i}$ such that $\sum_{\substack{\coalition' \neq \coalition\\\coalition' \ni i}} \network^n_{\coalition'} + \networkbis_\coalition < \budget_i$ for all $n \geq \bar n_{3,i}$.

	For every $i \in \agents^=_\coalition$, since $\sum_{\substack{\coalition' \neq \coalition\\\coalition' \ni i}} \network_{\coalition'} + \networkbis_\coalition = \budget_i$ and $\network_\coalition < \networkbis_\coalition$, we have $\sum_{\substack{\coalition' \neq \coalition\\\coalition' \ni i}} \network_{\coalition'} + \network_\coalition < \budget_i$, which implies that there exists $\bar n_{3,i}$ such that $\network_\coalition < \budget_i - \sum_{\substack{\coalition' \neq \coalition\\\coalition' \ni i}} \network^n_{\coalition'}$ for all $n \geq \bar n_{3,i}$.
		
	Let $\bar n = \max \{\bar n_1,\max \{\bar n_{2,i} : i \in \agents\},\max \{\bar n_{3,i} : i \in \agents\}\}$, define $\network^{n_m} = \network^{\bar n + m}$, and set
		$$(\network')^m = (\hat \networkbis^{n_m}_\coalition,\network^{n_m}_{-\coalition}).$$
		
	We first check that, for every $m$, $(\network')^m$ belongs to $\networksab$:
	\begin{itemize}
		\item Since $\networkbis_\coalition \leq \capacity_\coalition$, we immediately have
		$$\min \{\networkbis_\coalition, \min \{\budget_i - \sum_{\substack{\coalition' \neq \coalition\\\coalition' \ni i}} \network^{n_m}_{\coalition'} : i \in \agents^=_\coalition\}\} \leq \capacity_\coalition.$$
		\item Let $j \in \agents$.
		\begin{itemize}
			\item If $j \in \coalition$, then, by assumption,
			$$\sum_{\substack{\coalition' \neq \coalition\\\coalition' \ni j}} \network_{\coalition'} + \networkbis_\coalition \le \budget_j.$$
			Therefore,
			$$\sum_{\substack{\coalition' \neq \coalition\\\coalition' \ni j}} \network^{n_m}_{\coalition'} + \min \{\networkbis_\coalition, \min \{\budget_i - \sum_{\substack{\coalition' \neq \coalition\\\coalition' \ni i}} \network^{n_m}_{\coalition'} : i \in \agents^=_\coalition\}\} \leq \budget_j.$$
			\item If $j \notin \coalition$, feasibility follows directly from the fact that
			$$\sum_{\substack{\coalition' \in \edges\\ \coalition' \ni j}} \network^{n_m}_{\coalition'} \le \budget_j,$$
			since the weights of coalitions not equal to $\coalition$ are unchanged. 
		\end{itemize}
	\end{itemize}
	Hence, $((\network')^m)_{m \in \mathbb{N}}$ is a sequence in $\networksab$.
		
	We now show that, for every $m$, $(\network')^m$ belongs to $\Psi(\network^{n_m})$.
	Consider the coalition $\coalition$ together with the weight $\hat \networkbis^{n_m}_\coalition$:
	\begin{itemize}
		\item By construction of $\bar n$, the following properties hold:
		\begin{itemize}
			\item $\networkbis_\coalition > \network^{\ell_k}_\coalition$ (since $\networkbis_\coalition > \network_\coalition$), and $\budget_i - \sum_{\substack{\coalition' \neq \coalition\\\coalition' \ni i}} \network^{n_m}_{\coalition'} > \network^{n_m}_\coalition$, for every $i \in \agents^=_\coalition$.
			Therefore,
			$$\hat \networkbis^{n_m}_\coalition > \network^{n_m}_\coalition.$$
			\item For every $i \in \coalition$, by continuity of $\payoff_i$,
			$$\payoff_i((\network')^m) > \payoff_i(\network^{n_m}).$$
			\item For every $i \in \agents^<_\coalition$, since $\sum_{\substack{\coalition' \neq \coalition\\\coalition' \ni i}} \network^{n_m}_{\coalition'} + \networkbis_\coalition < \budget_i$, we have
			$$\sum_{\substack{\coalition' \neq \coalition\\\coalition' \ni i}} \network^{n_m}_{\coalition'} + \hat \networkbis^{n_m}_\coalition < \budget_i.$$
		\end{itemize}
		\item For every $i \in \agents^=_\coalition$, since $\sum_{\substack{\coalition' \neq \coalition\\\coalition' \ni i}} \network^{n_m}_{\coalition'} + (\budget_i - \sum_{\substack{\coalition' \neq \coalition\\\coalition' \ni i}} \network^{n_m}_{\coalition'}) = \budget_i$, we have
		$$\sum_{\substack{\coalition' \neq \coalition\\\coalition' \ni i}} \network^{n_m}_{\coalition'} + \min \{\networkbis_\coalition, \min \{\budget_j - \sum_{\substack{\coalition' \neq \coalition\\\coalition' \ni j}} \network^{n_m}_{\coalition'} : j \in \agents^=_\coalition\}\} = \sum_{\substack{\coalition' \neq \coalition\\\coalition' \ni i}} \network^{n_m}_{\coalition'} + \hat \networkbis^{n_m}_\coalition \leq \budget_i.$$
	\end{itemize}
	Hence, the point $(\network')^m$ satisfies condition $(2)$ in the definition of $\Psi(\network^{n_m})$, thus $(\network')^m \in \Psi(\network^{n_m})$.
		
	\paragraph{Step 2. For every $\network \in \networksab$, $\network \notin \mathrm{Conv}(\Psi(\network))$}
		
	Assume by contradiction that $\network \in \mathrm{Conv}(\Psi(\network))$.
		
	Then, there exists $(\network^n)_{n \in M} \subset \Psi(\network)$ and weights $(\lambda^n)_{n \in M}$ such that
	$$\network = \sum_{n \in M} \lambda^n \network^n.$$
	Fix a coalition $\coalition \in \edges$ such that $\network^n_\coalition \neq \network_\coalition$, for some $n$.
	Let $M_\coalition = \{n \in M : \network^n_\coalition \neq \network_\coalition\}$.
	Then $\network_\coalition$ is a convex combination of $(\network^n_\coalition)_{n \in M_\coalition}$.
		
	There must exist $\bar n_1 \in M_\coalition$ such that $(\network^{\bar n_1}_\coalition,\network_{-\coalition})$ satisfies condition $(1)$, and $\bar n_2 \in M_\coalition$ such that $(\network^{\bar n_2}_\coalition,\network_{-\coalition})$ satisfies condition $(2)$; otherwise, quasiconcavity of $\payoff_i$ ($i \in \coalition$) yields a contradiction.
		
	Let $i \in \coalition$ such that $\payoff_i(\network^{\bar n_1}_\coalition,\network_{-\coalition}) > \payoff_i(\network)$.
	Note that $\payoff_i(\network^{\bar n_2}_\coalition,\network_{-\coalition}) > \payoff_i(\network)$.
	Since $\network_\coalition$ lies strictly between $\network^{\bar n_1}_\coalition$ and $\network^{\bar n_2}_\coalition$, quasiconcavity of $\payoff_i$ implies that
	$$\payoff_i(\network) \geq \min \{\payoff_i(\network^{\bar n_1}_\coalition,\network_{-\coalition}),\payoff_i(\network^{\bar n_2}_\coalition,\network_{-\coalition})\} > \payoff_i(\network),$$
	a contradiction.
		
	Hence, $\network \notin \mathrm{Conv}(\Psi(\network))$.
		
	\paragraph{Step 3. Application of Yannelis--Prabhakar theorem.}
		
	From Steps 1 and 2, all assumptions of Theorem \ref{thm_yannelis_prabhakar} are satisfied.
	Therefore, there exists $\network^\ast \in \networksab$ such that $\Psi(\network^\ast) = \emptyset$.
	This implies exactly the two stability conditions for $\network^\ast$.
    \hfill $\qed$

	\subsection{Proof of \texorpdfstring{\Cref{lemma:stable_matching}}{}} \label{app.proof:stable_matching}
    
    Let $\network$ be a stable matching of a fractional matching problem $\game = (\agents,\edges,\payoff, \capacity, \budget)$ with $\payoff(\network) = (\sum_{\coalition \in \edges^i}A^i_\coalition x_\coalition)_{i \in \agents}$. No agent has an incentive to decrease one of their weights since all weights for which they have a negative coefficient are 0. For any coalition $\coalition$ in which all members have a positive coefficient, one member $i$ verifies $\sum_{\coalition' \ni i, A^i_{\coalition'} \geq A^i_\coalition} \network_{\coalition'} = \budget_i$, so increasing $\network_\coalition$ while decreasing other weights of $i$ cannot strictly increase $i$'s payoff.
      
    On the other hand, let $\network$ be a fully stable network of a fractional matching problem $\game = (\agents,\edges,\payoff, \capacity, \budget)$ with $\payoff(\network) = (\sum_{\coalition \in \edges^i}A^i_\coalition x_\coalition)_{i \in \agents}$ and $\vert \coalition \vert = 2, \forall \coalition \in \edges$. If some weight $\network_\coalition >0$ and $\exists i \in \coalition$, $A^i_\coalition <0$, then $i$ would have an incentive to decrease $\network_\coalition$, which contradicts full stability. If there is an edge $\coalition = \{i, j\} \in \edges$ such that $\forall k \in \coalition$, $A^k_\coalition>0$ and $\sum_{\coalition' \ni k, A^k_{\coalition'} \geq A^k_\coalition} \network_{\coalition'} < \budget_k$, then $\exists k, \ell \in \agents$ such that $\{i, k\}, \{j, \ell\} \in \edges$, $A^i_{\{i, k\}}<A^i_{\{i, j\}}$, $A^j_{\{j, \ell\}}<A^j_{\{i, j\}}$, $\network_{\{i, k\}} > 0$ and $\network_{\{j, \ell\}} > 0$. Then, both $i$ and $j$ have an incentive to increase $\network_\coalition$ by some value $\varepsilon$ while decreasing $\network_{\{i, k\}}$ and $\network_{\{j, \ell\}}$ by $\varepsilon$, which contradicts full stability. \hfill $\qed$

    \subsection{Proof of \texorpdfstring{\Cref{prop:min_argmax}}{}} \label{app.proof:min_argmax}

    Let $\network \in \networks$ that satisfies equation \eqref{eq:min_argmax}.
    Let $\coalition \in \edges$. By quasiconcavity and the definition of $\network_\coalition$, $\payoff_i(\cdot, x_{-\coalition})$ is non-decreasing on $[0, \network_\coalition]$ for all $i \in \coalition$ and non-increasing on $[ \network_\coalition, 1]$ for some $i\in \coalition$. As a consequence, no agent can increase their payoff by unilaterally decreasing one of their weight, and no coalition $\coalition$ can increase the payoff of all its members by increasing increasing $\network_\coalition$ (and maybe decreasing other weights): indeed, the agent for whom $\payoff_i(\cdot, x_{-\coalition})$ is non-increasing on $[ \network_\coalition, 1]$ does not gain from the increase of $\network_\coalition$, neither from a decrease of weights $\network_{\coalition'}, i\in \coalition'$ because $\payoff_i(\cdot, x_{-\coalition'})$ is non-decreasing on $[0, \network_{\coalition'}]$, and they cannot gain anything from a decrease of a weight $\network_{\coalition'}, i\notin \coalition'$ because of the positive externalities assumption.

    Let $\network\in \networks$ and assume $\network$ does not satisfy equation \eqref{eq:min_argmax}. By quasiconcavity, $\forall \coalition \in \edges, \forall i \in \coalition$, $\argmax \payoff_i(\cdot, x_{-\coalition})$ is a closed interval. Then, $\exists \coalition \in \edges$ such that either
    \begin{enumerate}
        \item $\exists i \in \coalition$, $x_\coalition > \argmax \payoff_i(\cdot, x_{-\coalition})$, or
        \item $\forall i \in \coalition$, $x_\coalition < \argmax \payoff_i(\cdot, x_{-\coalition})$.
    \end{enumerate}
        
    Case (1): $\network$ cannot be stable since $i$ has an incentive to decrease $\network_\coalition$.

    Case (2): $\network$ cannot be stable since all members of $\coalition$ have an incentive to increase $\network_\coalition$.

    Therefore, the fully stable networks are exactly those that satisfy equation \eqref{eq:min_argmax}.
    Finally, if each agent's payoff $\payoff_i$ is component-wise \emph{strictly} quasiconcave, $\argmax \payoff_i(\cdot, x_{-\coalition})$ is reduce to a single point for all $\coalition \in \edges$ and for all $i \in \coalition$, so there is a unique fully stable network. \hfill $\qed$

    \subsection{Proof of \texorpdfstring{\Cref{thm:bipartite_full}}{}}\label{app.proof:bipartite_full}

    \paragraph{Part 1}
        
    The quantity $\sum_{j \in V}\payoff_j(\network) - \sum_{\coalition\in\edges} \capacity_C$ is strictly increasing throughout the algorithm. Indeed, $V$ agents either keep their weights (and their payoff) constant, or their change some weights to strictly increase their payoff, and if no weight has changed it means that either the algorithm has terminated or some capacities have changed. Moreover, $\sum_{j \in V}\payoff_j(\network) - \sum_{\coalition\in\edges} \capacity_C$ is upper bounded by $\sum_{j \in V} \max_{\network \in \networksab}\payoff_j(\network)$, so the algorithm converges (potentially in an infinite number of steps).

    \paragraph{Part 2}
    
    Let us write, by separability, $\payoff_i  = \sum_{\coalition \ni i}\payoff_i^\coalition$ for all $i \in \agents$. Let $i \in U$. Let $\networkbis^i\in\networks^i,~ \networkbis^i\preceq \network^i$. For any $\coalition$ such that $\network_\coalition >0$, by (quasi)concavity $\payoff_i^\coalition$ is non-decreasing on $[0, \network_\coalition]$: either $i$ chose to increase $\network_C$ to its current value, or they increased it further and it was later decreased by a $V$ agent. Therefore, $\payoff_i(\networkbis^i) \leq \payoff_i(\network^i)$, so there are no profitable unilateral deviations. 

    Let $\coalition = \{i, j\}\in\edges$, let $\networkbis^C \in \networks^C$ such that $\networkbis^\coalition_\coalition \geq \network_\coalition$, $\forall \coalition' \in \edges^\coalition \backslash \{\coalition\}$, $\networkbis^\coalition_{\coalition'} \leq \network_{\coalition'}$, and $(\networkbis^C, \network^{-C})\in \networksab$. If $\payoff_i(\networkbis^C) \leq \payoff(\network)$, $\networkbis^C$ is not a mutually profitable deviation. Assume $\payoff_i(\networkbis^C) > \payoff(\network)$, then, we must have $\capacity_C = \network_C$. Indeed, if $i$ does not saturate their budget, for the algorithm to terminate they must have proposed to increase $\network_C$ and be denied. If they do saturate their budget, then $\vert \networkbis^C \vert_1 \leq \vert \network^C \vert_1$, and since $\payoff_i(\networkbis^C) > \payoff(\network)$, $\exists \coalition'$ such that $\networkbis_{C'} < \network_{C'}$ and $\partial^-\payoff^{C'}_i(\network_{C'}) < \partial^+\payoff^{C}_i(\network_{C})$\footnote{We denote by $\partial^-$ and $\partial^+$ respectively the left and the right derivatives.}. Consider the last step at which $\network_{\coalition'}$ was increased, and let us call $z$ the state of the algorithm at the start of the step. Four cases are possible:
    \begin{enumerate}
        \item $i$ tried to increase $z_{\coalition}$ but was denied, or
        \item $i$ succeeded to increase $z_{\coalition}$, or
        \item $i$ did not try to increase $z_{\coalition}$ because it was equal to $\capacity_C$, or
        \item $i$ did not try to increase $z_{\coalition}$ because it was not in the support of the argmax.
    \end{enumerate}
    
    Case (1): Then $\capacity_C$ was set to $z_C$, and from then on this weight could only be decreased by $j$ and $\capacity_C$ with it, so $\capacity_C = \network_C$.
    
    Case (2): $i$ increased both $z_C$ and $z_{C'}$. Either they increased $z_C$ until $\capacity_C$, and by the same argument as before $\capacity_C = \network_C$, or they stopped at some point $z'_C<\capacity_C$, which means that $\partial^+\payoff^{C'}_i(z_{C'}) > \partial^+\payoff^{C}_i(z'_{C})$. Since we know that $ \partial^+\payoff^{C}_i(\network_{C}) > \partial^-\payoff^{C'}_i(\network_{C'}) \geq \partial^+\payoff^{C'}_i(z_{C'})$, necessarily $\network_C < z'_C$, which implies that $j$ has decreased it at some step and thus $\capacity_C = \network_C$.

    Case (3): Nothing to do.

    Case (4): If $z_{\coalition}$ was not in the support of the argmax, then $\network_C < z_C$ and by the same argument $\capacity_C = \network_C$.

    Now that we have $\capacity_C = \network_C$, consider the step at which $j$ set $\capacity_C = \network_C$. At the end of that step, their budget constraint was saturated, and their weight vector $z^j$ was such that $z_C = \network_C$ and for all $\coalition' \in \text{support}(z^j)$, $\partial^-\payoff_j^{\coalition'}(z_{\coalition'}) \geq \partial^+\payoff_j^{\coalition}(\network_{\coalition})$. This inequalities must remain true until the end of the algorithm: $\network_C$ remains constant, if $z_{\coalition'}$ decreases $\partial^-\payoff_j^{\coalition'}(z_{\coalition'})$ decreases also, and $j$ only has an incentive to increase $z_{\coalition'}$ to some value $z'_{\coalition'}$ if $\partial^-\payoff_j^{\coalition'}(z'_{\coalition'}) \geq \partial^+\payoff_j^{\coalition}(\network_{\coalition})$. Therefore, no deviation that consists in increasing $\network_C$ can be profitable to $j$. In conclusion, no agent or pair of agents have any profitable deviations, so $\network$ is fully stable. \hfill $\qed$
    
\end{document}